\documentclass[11pt]{article}
\usepackage{fullpage}
\usepackage{url}
\usepackage[colorlinks=true,linkcolor=Emerald,citecolor=RoyalBlue]{hyperref}
\usepackage{amsmath,nccmath,amsfonts,amsthm,multirow,mdframed,amssymb}

\usepackage{times}
\usepackage{paralist}
\usepackage{graphicx}
\usepackage{floatpag}
\usepackage{float}

\usepackage{bbold}
\usepackage{enumitem}
\usepackage{subcaption} 
\usepackage{soul}
\usepackage{comment}
\usepackage{calc}
\usepackage{upgreek}
\usepackage{mathtools}
\usepackage{color}
\usepackage[dvipsnames]{xcolor}
\usepackage{xspace}
\usepackage{tcolorbox}
\usepackage{cleveref}
\usepackage{stackengine}
\usepackage{tocloft}
\usepackage{physics}
\usepackage{quantikz}

\stackMath
\allowdisplaybreaks

\newtheorem{theorem}{Theorem}
\newtheorem*{theorem*}{Theorem}
\newtheorem{lemma}[theorem]{Lemma}
\newtheorem{corollary}[theorem]{Corollary}
\newtheorem{fact}[theorem]{Fact}

\crefname{fact}{Fact}{Facts}
\Crefname{fact}{Fact}{Facts}

\newcommand{\ad}{\operatorname{ad}}
\newcommand{\lr}[1]{\left(#1\right)}
\newcommand{\wt}{\widetilde}

\newcommand{\E}{\mathbb E}

\newcommand{\cB}{\mathcal{B}}

\newcommand{\cD}{\mathcal{D}}
\newcommand{\cE}{\mathcal{E}}

\newcommand{\cH}{\mathcal{H}}

\newcommand{\cK}{\mathcal{K}}
\newcommand{\cL}{\mathcal{L}}
\newcommand{\cM}{\mathcal{M}}

\newcommand{\cP}{\mathcal{P}}

\newcommand{\C}{\mathbb C}
\newcommand{\R}{\mathbb R}

\newcommand{\sM}{\mathsf M}

\numberwithin{equation}{section}

\makeatletter
\let\c@fconjecture\c@conjecture
\makeatother

\makeatletter
\let\c@fconj\c@conj
\makeatother

\title{\fontsize{15.5}{19}\selectfont Fast mixing of the SYK model at high temperatures}
\author{
\begin{tabular}{c@{\hspace{1.1cm}}c@{\hspace{0.8cm}}c}
Yiyi Cai &
Yongtao Zhan&
Alexander Zlokapa\\[2pt]
\small Stanford &
\small Caltech &
\small MIT
\end{tabular}
}
\date{}

\begin{document}
\maketitle
\thispagestyle{empty}
\begin{abstract}
The Sachdev-Ye-Kitaev (SYK) model is a strongly interacting fermionic system. Its dynamics are difficult to analyze with Lieb-Robinson bounds and cluster expansions due to its all-to-all disordered interactions. We prove that, with high probability over the disorder, the SYK model admits a quantum Gibbs sampler with a system-size-independent spectral gap at sufficiently high constant temperatures. This yields polynomial-time preparation of SYK Gibbs states on a quantum computer. While prior non-rigorous computations that suggested SYK thermalization were limited to studying local correlators, we emphasize that our result holds up to arbitrary inverse polynomial trace distance. Our proof introduces the pseudo-Lindbladian approach to fermions then applies two key new ideas: a Bochner inequality for quantum dynamics, and an especially simple SYK analysis based on a comparison to classical dynamics.
\end{abstract}

\section{Introduction}
Understanding strongly interacting fermionic systems is a central challenge in quantum algorithms. Despite their importance in condensed matter and quantum chemistry, they are historically challenging to simulate classically due to obstructions such as the fermionic sign problem and large entanglement, obstructing Quantum Monte Carlo \cite{troyer2005computational} and tensor network methods \cite{mortier2025fermionic,kraus2010fermionic}. One of the most well-studied strongly interacting systems is the Sachdev-Ye-Kitaev (SYK) model \cite{sachdev1993gapless,maldacena2016remarks,2015escq.progE...2K,kitaev2018soft}. It is a disordered mean-field model on $N$ Majorana fermions, with independent Gaussian coefficients $J_A \stackrel{\text { iid }}{\sim} \mathcal{N}\left(0,6 J^2 / N^3\right)$ coupling every set of four distinct Majorana fermions $\{\gamma_{i}\}$ (defined as Hermitian operators satisfying $\{\gamma_i, \gamma_j\} = 2\delta_{ij}$):
\[
H=\sum_{1 \leq i_1 < \cdots < i_4 \leq N} J_{i_1 i_2 i_3 i_4} \gamma_{i_1} \gamma_{i_2} \gamma_{i_3} \gamma_{i_4}.
\]
Rigorous results show that, with high probability over the disorder, preparing thermal states at any constant temperature requires quantum circuits whose depth grows polynomially with system size; moreover, these states are far from Hartree-Fock or Gaussian approximations \cite{hastings2022optimizing,hastings2023field,anschuetz2025strongly,ramkumar2026high, zlokapa2026rigorous}. Although recent results prove that \emph{local} thermal expectations are classically efficient to estimate~\cite{zlokapa2026rigorous,schmidhuber2026polynomial}, the \emph{global} state itself still appears to be far from classical. Hence, it remains open if physically preparing the SYK Gibbs state is a possible route towards demonstrating quantum advantage.

Non-rigorous results using techniques from statistical physics suggest that the SYK Gibbs state should be quantumly easy to prepare, e.g., by coupling it to a large thermal bath~\cite{zhang2019evaporation,maldacena2021syk,almheiri2024universal,schuster2025cooling}. These results use methods that fundamentally cannot analyze how close the resulting state is in \emph{trace distance}: they only address properties of local observables. Hence, while open system dynamics may eventually reproduce quantities such as thermal two-point functions, it is (even non-rigorously) open if global quantities (e.g., a Helstrom measurement to certify trace distance) can be correctly reproduced after thermalizing. Unfortunately, rigorous analyses of the SYK model are often arduous and highly technical, and have only recently confirmed some longstanding predictions from the physics literature~\cite{chen2021growth,lucas2020non,zlokapa2026rigorous,gamarnik2026free,he2026spectral,schmidhuber2026polynomial}.

Recent quantum Gibbs sampling algorithms offer a rigorous path to efficient thermal state preparation. They are often based on dissipative Lindbladian dynamics that are carefully designed to satisfy detailed balance and are thus natural candidates for algorithms that imitate natural thermalization processes~\cite{gilyen2024quantum, chen2023quantum, chen2023efficient, chen2025efficient, jiang2024quantum}. Similar to classical Markov chain Monte Carlo, the efficiency of the algorithm crucially depends on the mixing time: the characteristic timescale for dynamics to approach the stationary state, regardless of the initial state. For a quantum sampler, we measure convergence in trace distance, so a mixing time guarantee controls the outcome probabilities of arbitrary measurements, including global observables. Establishing such mixing time bounds is often technically challenging: even classically, the mixing times of mean-field disordered systems have only recently been understood (via spectral independence)~\cite{anari2024universality,anari2024trickle}. In the quantum setting, a spectral gap of the sampler can be used to show a mixing time bound \cite{temme2010chi}, while stronger functional inequalities can be used to improve its dependence on system size \cite{kastoryano2013quantum}.
Recent proofs of quantum fast mixing for noncommuting systems have used Lieb-Robinson bounds and spectral-gap stability~\cite{rouze2025efficient,rouze2026optimal}, quantum Dobrushin conditions~\cite{bakshi2026dobrushin}, and most recently pseudo-Lindbladian comparisons~\cite{bergamaschi2026fast}. Arguments based on Lieb-Robinson bounds have also established fast mixing of weakly interacting fermionic systems \cite{h1dx-ps5p, vsmid2025polynomial, vsmid2025rapid}. 

However, the SYK model is an \emph{all-to-all} system, where each fermion strongly interacts with all other fermions in the system, rather than only within a local neighborhood. Although each individual coupling is weak, the number of interactions involving a given fermion grows rapidly with system size, on the order of $O(N^3)$. Simply adding their absolute strengths then gives a bound that grows with the system size. However, such an estimate ignores a specific feature of the SYK model: the couplings between fermions have random signs, causing cancellations. Exploiting this randomness is essential for obtaining useful bounds at a temperature independent of system size. Standard Lieb-Robinson bounds are too loose for SYK, and it is unclear if more refined bounds on operator growth that exploit the random interactions~\cite{chen2021growth,lucas2020non} suffice to establish fast mixing at constant temperature.
An alternative strategy was recently developed that can show fast mixing for systems when Lieb-Robinson bounds fail~\cite{bergamaschi2026fast}. These techniques, based on cluster expansions, yield fast mixing at high temperatures for all-to-all models~\cite{bergamaschi2026fastall}. However, when applied directly to SYK, their sufficient temperature threshold grows with system size. It is thus unclear if any of the existing methods of bounding mixing times are powerful enough to apply to the SYK model.

In this work, we prove fast mixing of the SYK model at large constant temperature. This implies a polynomial-time quantum algorithm for preparing the state up to arbitrary inverse polynomial trace distance.

\begin{theorem}[Fast mixing of SYK]\label{thm:main_result}
Let $H$ be the Hamiltonian of the SYK model, where $\gamma_i$ are Majorana operators satisfying the anticommutator relations $\left\{\gamma_j, \gamma_k\right\}=2 \delta_{j k}$. Then for all $0 \leq \beta\leq 1/(1000J)$, there is at least a probability $1-2e^{-N}$ over the random couplings that, the Lindbladian defined in \eqref{eq:lindblad} admits a system-size-independent gap, and the mixing time of the Lindbladian is at most $O(N+\log (1/{\epsilon}))$.
\end{theorem}
\noindent
To obtain an efficient state-preparation algorithm, we combine this mixing-time bound with the efficient implementation of the Gibbs sampler \cite{ding2025efficient}. The resulting procedure prepares the SYK Gibbs state to inverse-polynomial trace-distance error in polynomial time.

Our proof combines two new techniques. First, we extend the pseudo-Lindbladian framework \cite{bergamaschi2026fastall, bergamaschi2026fast} to a more convenient form for fermionic systems. To bound its spectral gap, we use a fermionic version of the \textit{integrated Bochner method}, which bounds the square of a nonnegative generator by its Dirichlet form. This strategy underlies classical Bakry-\'Emery arguments \cite{bakry1985diffusions, boudou2006spectral, kondratiev2013spectral} and recent classical mixing-time proofs~\cite{gobel2026simple, chen2026spectral2, chen2026spectral}. We carry out the calculation using Clifford derivatives, which are the fermionic counterparts of ordinary derivatives~\cite{gross1975clifforddirichlet,carlen2014fermionic}. (Whereas ordinary derivatives commute, Clifford derivatives anticommute.) These relations allow us to obtain a criterion expressed in terms of the Jacobian of imaginary-time evolution. Roughly speaking, it sufficies to show that the block matrix $D\gamma(s)$ with elements $(D\gamma(s))_{ji} = \frac{1}{2}\{\gamma_i, e^{-sH} \gamma_j e^{sH}\}$ is close to the identity. This matrix is precisely the Jacobian defined with respect to the Clifford derivative, used in fermionic Dirichlet forms and Fokker-Planck dynamics. This Jacobian measures how a perturbation at site $i$ affects the imaginary-time-evolved fermion at site $j$. To relate $e^{-sH} \gamma_j e^{sH}$ back to the original jump operators used in algorithms like~\cite{ding2025efficient}, we show that it suffices to similarly analyze $\partial_s D\gamma(s)$.

The second key idea is an analysis of the SYK model that hinges on a comparison to a \emph{classical} spin glass. For a Hamiltonian $\sum_{i_1<\cdots<i_4} J_{i_1\cdots i_4} x_{i_1}\cdots x_{i_4}$ defined on commuting variables $x_{i_k}$ with the same disorder as the original SYK model, we show that the classical Jacobian $(Dx(s))_{ji} = \partial_i x_j(s)$, evaluated at a uniformly random Ising spin configuration, can be used to bound the SYK Jacobian $D\gamma(s)$. The proof relies on the classical dynamics having the same form as the fermionic dynamics, up to signs introduced by reordering fermions; these signs are bounded away by the triangle inequality. This is a noticeably simpler approach than prior rigorous SYK analyses based on the Lovasz theta function~\cite{anschuetz2025strongly}, mean-field cluster expansions~\cite{zlokapa2026rigorous}, or the cavity method~\cite{gamarnik2026free,he2026spectral,schmidhuber2026polynomial}.

\subsection{Proof sketch}
We study the Lindbladian construction of Ding, Li, and Lin ~\cite{ding2025efficient}, in the line of efficiently implementable quantum Gibbs samplers with exact Kubo-Martin-Schwinger (KMS) detailed balance \cite{chen2023efficient} for preparing the Gibbs state $\rho = e^{-\beta H}/\Tr(e^{-\beta H})$. KMS detailed balance is the quantum reversibility condition used here, which ensures the correct stationary state but does not by itself bound the mixing time. We build jump operators from Majoranas $\gamma_i$ in order to later exploit their algebra and make the gap accessible through Clifford derivatives.  The construction uses real-time evolution with a Gaussian filter to prepare the jump operators
\begin{align}
    A_i = e^{-\beta H/4}\left[e^{-\beta^2\operatorname{ad}_H^2/8}(\gamma_i)\right]
    e^{\beta H/4} \quad \text{for} \quad \ad_H(X) = [H, X].
\end{align}
Together with the coherent correction in \eqref{eq:lindblad}, these jumps define the Gibbs sampler we use in this paper.
To prove fast mixing for these dynamics, we use a proof technique based on the pseudo-Lindbladian construction of~\cite{bergamaschi2026fast}, which relates the spectral gap of the Gibbs sampler to that of an auxiliary operator built from the same jumps. Specifically, they showed (up to sign and normalization conventions) that the gap of the pseudo-Lindbladian
\begin{align}
    \cK_f(X)
    =\frac{1}{4} \sum_{i=1}^N
      \left[A_i^\dagger(A_iX-XA_i^\dagger)-(A_iX-XA_i^\dagger)A_i\right]
\end{align}
implies the gap of the Gibbs sampler. Unfortunately, proving a gap of $\cK_f$ can be challenging for SYK, whose dense, disordered interactions are not directly controlled by the Lieb-Robinson estimates \cite{rouze2025efficient, h1dx-ps5p, vsmid2025polynomial, vsmid2025rapid} or cluster-expansion bounds \cite{bergamaschi2026fastall} used in prior mixing-tiume analyses. 

\paragraph{Classical Bochner argument.} To explain the spectral-gap argument for SYK, we first consider its classical analogue. Denote the Gibbs measure
\begin{align}
    d\mu (x) = Z^{-1} e^{-U(x)} dx
\end{align}
on $\R^n$. For real-valued functions, write 
\begin{align}
    \langle f, g\rangle_{\mu} = \int fg \, d\mu, \qquad \|f\|_{2,\mu}^2=\langle f,f\rangle\mu=\int f^2\,d\mu, \qquad \mu(f)=\int f\,d\mu.
\end{align}
We consider overdamped Langevin dynamics, whose generator can be written in terms of the derivatives $\partial_i, \partial_i^\dagger$. The adjoint $\partial_i^\dagger$ of the coordinate derivative $\partial_i$ is taken with respect to this weighted inner product and satisfies 
\begin{align}\label{eq:adjoint_relation}
    \langle \partial_i f, g \rangle_{\mu} = \langle f, \partial_i^{\dagger} g \rangle_{\mu}.
\end{align}
Integrating by parts gives 
\begin{align}
    \partial_i^\dagger = - \partial_i + (\partial_i U), \quad [\partial_i, \partial_j^\dagger] = \partial_i \partial_j U,
\end{align}
and these identities hold for smooth functions with suitable decay at infinity.
The positive generator $L$ (so that the semigroup is $e^{-tL}$) and its Dirichlet form are 
\begin{align}
    L = \sum_i \partial_i^\dagger \partial_i = - \Delta + \nabla U \cdot \nabla,\quad  \cE(f,g) := \langle f, Lg \rangle_{\mu} = \int \nabla f \cdot \nabla g \,du.
\end{align}
The key step is to expand $L^2$. Using the commutation of ordinary derivatives gives 
\begin{align}\label{eq:L_squared}
    L^2 = \sum_{i,j} \partial_i^\dagger [\partial_i, \partial_j^\dagger] \partial_j + \sum_i \partial_i^\dagger L \partial_i.
\end{align}
Taking the inner product of ~\eqref{eq:L_squared} with $f$ and using the adjoint relation in ~\eqref{eq:adjoint_relation} give the integrated Bochner identity 
\begin{align}
    \int (Lf)^2 d\mu = \sum_i \langle \partial_i f, L \partial_i f\rangle_{\mu} + \int (\nabla f)^T (\nabla^2U) \nabla f \, d\mu.
\end{align}
The first term is nonnegative since it is $\int \|\nabla^2 f\|_{\text{HS}}^2 d\mu$ where $\|\cdot\|_{\text{HS}}$ denotes the Hilbert-Schmidt norm of the matrix. Therefore, if $\nabla^2 U(x) \geq \kappa I$ for every $x$, with $\kappa >0$, then the second term is at least $\kappa \int |\nabla f|^2 d\mu$. Then, 
\begin{align}\label{eq:integrated_inequality}
    \int (Lf)^2 d\mu \geq \kappa \int |\nabla f|^2 d\mu = \kappa \langle f, Lf\rangle_{\mu}.
\end{align}
Since $L$ is self-adjoint, the left-hand side equals $\langle f, L^2 f\rangle_\mu$, so 
\begin{align}
    \langle f, (L^2 - \kappa L) f\rangle_{\mu} \geq 0
\end{align}
for every suitable function $f$. This implies that every nonzero spectral value $\lambda$ satisfies $\lambda^2 \geq \kappa \lambda$ so $\lambda \geq \kappa$, which gives us a spectral gap of at least $\kappa$, or equivalently $\text{Var}_\mu (f) \leq \frac{1}{\kappa} \int |\nabla f|^2 d\mu$. The spectral gap controls how quickly the dynamics approach equilibrium:
\begin{align}
    \|e^{-Lt} f - \mu(f)\|_{2, \mu} \leq e^{-\kappa t} \|f - \mu(f)\|_{2,\mu}
\end{align}
The Hessian condition $\nabla^2 U \geq \kappa I$ gives the spectral-gap above through the Bakry-\'Emery \cite{bakry1985diffusions}, and Bochner identities give analogous arguments for interacting particle systems \cite{boudou2006spectral, kondratiev2013spectral}. More fundamentally, the integrated inequality in ~\eqref{eq:integrated_inequality} suffices for every test function, even when the point-wise Hessian condition is unavailable. 
The Hessian in this calculation can also be viewed as a Jacobian: for $F=\nabla U$, we have $DF=\nabla^2U$. For example, a uniform bound $|DF(x)-I|\leq\delta<1$ implies $\nabla^2U(x)\geq(1-\delta)I$, and hence a spectral gap of at least $1-\delta$. Our fermionic argument follows the same idea, where we replace ordinary derivatives by Clifford derivatives, and the square expansion involves anticommutators of these derivatives with their adjoints in place commutators. The resulting block operator plays the role of the classical Hessian, which, as discussed below, can be expressed in terms of the Jacobian $D\gamma(s)$ of imaginary-time evolution, so controlling this Jacobian gives a spectral-gap bound.

\paragraph{Pseudo-Lindbladians for anticommuting systems.} For fermionic systems, our first observation is that it is easier to avoid $A_i$ and instead study
\begin{align}
    \gamma_i(\beta/4) = e^{-\beta H/4}\gamma_i e^{\beta H/4}
\end{align}
since they exactly satisfy the anticommutation relation $\{\gamma_i(\beta/4), \gamma_j(\beta/4)\} = 2\delta_{ij}$, while the original jump operators $A_i$ do not. We thus use an \emph{unfiltered} pseudo-Lindbladian $\cK$ that replaces $A_i$ with $\gamma_i(\beta/4)$ compared to the definition of $\cK_f$, inspired by~\cite{bakshi2026dobrushin,bergamaschi2026fast,bergamaschi2026fastall}. To show fast mixing, we will need to show that $\cK$ is gapped, and that this gap can be transferred to $\cK_f$ by checking that
\begin{align}\label{eq:ccomp}
    \sum_i \norm{\Big(\gamma_i(\beta/4) X - X \gamma_i(-\beta/4)\Big) - \Big(A_i X - X A_i^\dagger\Big)}_2^2
\end{align}
is small. To analyze $\cK$, we further decompose $X$ by fermionic parity: the parity operator $\cP$ satisfies $\cP(X) = X$ for even operators and $\cP(X) = -X$ for odd operators, we show a gap for
\begin{align}
    \widetilde{\mathcal K}=\sum_i \wt\partial_i^\dagger \wt\partial_i \quad \text{for} \quad \wt\partial_i(X)=\frac12\Big(\gamma_i(\beta/4)X-\mathcal P(X)\gamma_i(-\beta/4)\Big)
\end{align}
similarly to the classical Bochner argument described above.
We use the notation $\wt\partial_i$ to be suggestive: up to a similarity transform $V(X) = \rho^{1/4} X \rho^{1/4}$, $\wt\partial_i = V\partial_i V^{-1}$ is the Clifford derivative
\begin{align}
    \partial_iX&=\frac12\left(\gamma_iX-\cP(X)\gamma_i\right), \qquad \partial_i(XY) = (\partial_i X)Y + \cP(X) \partial_i Y
\end{align}
studied, for example, in fermionic Fokker-Planck dynamics~\cite[Section~2.3]{carlen2014fermionic}. To show that $\wt\cK$ is gapped, we decompose its square as
\begin{align}\label{eq:k2}
    \wt\cK^2 = \sum_{i,j} \wt\partial_i^\dagger \{\wt\partial_i, \wt\partial_j^\dagger\}\wt\partial_j + \sum_i \wt\partial_i^\dagger \wt\cK \wt \partial_i.
\end{align}
The second term is already nonnegative; the contribution from the first term gives
\begin{align} \label{eq:fermion_squared_bound}
    \langle X, \wt \cK^2 X \rangle \geq \sum_{i,j} \langle \wt \partial_i X, \{\wt\partial_i, \wt\partial_j^\dagger\}(\wt\partial_j X)\rangle.
\end{align}
Compared to the classical Bochner calculation above, the anticommutators $\{\wt\partial_i, \wt\partial_j^\dagger\}$ replace the classical commutator $[\partial_i, \partial_j^\dagger].$ If the block operator satisfies 
\begin{align}
    \left[\{\wt\partial_i, \wt\partial_j^\dagger\}\right]_{i,j} \geq \kappa I,
\end{align}
then ~\eqref{eq:fermion_squared_bound} implies 
\begin{align}
    \langle X, \wt \cK^2 X \rangle \geq \kappa \sum_i \|\wt\partial_i X\|_2^2 = \kappa \langle X, \wt\cK X\rangle.
\end{align}
This gives a spectral gap of at least $\kappa$, exactly as in the classical argument. To establish the required lower bound, we express the anticommutators in terms of the imaginary-time-evolved Majoranas: 
\begin{align}
    4\{\wt\partial_i, \wt\partial_j^\dagger\}(X) = \{\gamma_i(\beta/4), \gamma_j(-\beta/4)\} X + X\{\gamma_j(\beta/4), \gamma_i(-\beta/4)\}. 
\end{align}
Therefore, it suffices to bound below the block matrix with entries 
\begin{align}
    \wt{D\gamma}_{ij} = \frac{1}{2}\{\gamma_i(\beta/4), \gamma_j(-\beta/4)\}.
\end{align}

If $\wt{D\gamma} \geq \kappa I$, then ${\rm gap}(\wt\cK) \geq \kappa$ is gapped. This matrix is related to the Jacobian via a similarity transform:
\begin{align}
    \wt{D\gamma}
    =(I_N\otimes e^{\beta H/4})D\gamma(\beta/2)(I_N\otimes e^{-\beta H/4}) \quad\text{for}\quad (D\gamma)_{ji}(s)&=\partial_i\gamma_j(s)=\frac12\{\gamma_i,\gamma_j(s)\}.
\end{align}
In particular, showing $\norm{D\gamma(\beta/2)-I} \leq 1/4$ gives $\wt{D\gamma} \geq (3/4)I$. 
The matrix $\wt{D\gamma}$ is Hermitian and has the same eigenvalues as $D\gamma$, so the norm bound places its spectrum in $[3/4, 5/4].$ Again, we can refer to a classical system to understand the Jacobian $D\gamma$. We will formally bound the quantum Jacobian $D\gamma$ in terms of a classical Jacobian to prove that $D\gamma$ is close to the identity.

\paragraph{Bounding SYK with a classical system.} While the preceding discussion is generically useful for analyzing fermionic mixing times, it does not address the core difficulties of rigorous SYK analyses. At each site $i$, the SYK model has $\sum_{A \ni i} |J_A|$ that grows with system size (with high probability). This obstructs standard cluster expansions, such as those used in~\cite{bergamaschi2026fastall} to show fast mixing of all-to-all systems. The Gaussian disorder must be taken into account: this averaging ensures that $\sum_{A \ni i} J_A^2 = O(1)$ is bounded. Unfortunately, cluster expansions~\cite{zlokapa2026rigorous} and Lieb-Robinson bounds~\cite{chen2021growth,lucas2020non} are extremely technical, and it is unclear if the existing versions of these methods suffice to obtain strong enough bounds. Here, we introduce a far simpler approach that compares the SYK model to a \emph{classical} disordered mean-field model. This allows us to address the disorder with standard probabilistic techniques to control $D\gamma$.

We consider a classical system defined with the same disorder as the SYK model, but defined on commuting variables $x_1,\dots,x_N$:
\begin{align}
    G(x) = \sum_{A\in \binom{[N]}{4}} J_A \prod_{i\in A} x_i.
\end{align}
Analogous to the fermionic imaginary-time evolution $\gamma_i(s) = e^{-sH}\gamma_i e^{sH}$, we consider the classical imaginary-time evolution
\begin{align}\label{eq:dsx}
    \frac{d}{ds}x(s; \sigma) = 2\nabla G(x(s;\sigma)), \qquad x(0;\sigma) = \sigma.
\end{align}
Here the initial condition $\sigma$ is an Ising configuration in $\{\pm1\}^N$, while the evolving coordinates take values in $\R^N$. These dynamics are deterministic and serve as an auxiliary comparison. We note that they are distinct from the Langevin sampling dynamics used in the classical Bochner argument.
The classical Jacobian 
\begin{align}
    (Dx(s;\sigma))_{ji} = \frac{\partial x_j(s;\sigma)}{\partial \sigma_i}
\end{align}
evaluated at the initial condition $\sigma$ measures the sensitivity of the evolved coordinate $x_j(s)$ to an initial coordinate $x_i(0) = \sigma_i$. Our main technical step compares the moments of the quantum and classical Jacobians in their Taylor expansion. For $D\gamma(s) = \sum_{m \geq 0} s^m F_m$ and $Dx(s;\sigma) = \sum_{m\geq 0}s^m \wt F_m(\sigma)$, we show that
\begin{align}\label{eq:jcomp}
    2^{-N/2} \E \Tr[\lr{F_m^\dagger F_m}^p] \leq 2^{-N} \sum_{\sigma\in\{\pm1\}^N} \E \Tr[\lr{\wt F_m(\sigma)^\dagger \wt F_m(\sigma)}^p],
\end{align}
where $\E$ denotes expectation over the Gaussian disorder. Thus, bounds on the classical coefficients yield moment bounds on the fermionic coefficients.  
The factors of $2^{-N/2}$ and $2^{-N}$ normalize the fermionic trace and the average over Ising configurations, respectively. By symmetry of the Gaussian disorder, the disorder-averaged expression on the right is the same for every $\sigma$, so we may tae $\sigma = (1, \dots, 1)$. 

The comparison follows from the similar action of the quantum and classical generators on products. For each factor in a product, an interaction containing its site replaces that factor with the three remaining sites in the interaction. For example, the interaction on $\{1,2,3,4\}$ replaces $x_1$ with $2J_{1234}x_2x_3x_4$, and the corresponding quantum interaction replaces $\gamma_1$ with $2J_{1234} \gamma_2\gamma_3\gamma_4$. The genereators therefore produce corresponding terms with the same couplings, except that the quantum terms may carry additional minus signs because Majorana operators anticommute. To see this, observe that the generator in the quantum evolution $\gamma_j(s) = e^{-s\ad_H} \gamma_j$ acts as
\begin{align}\label{eq:hw}
    [H, \gamma_{w_1}\cdots\gamma_{w_L}] = -2\sum_{a=1}^L \sum_{A\ni w_a} (-1)^{|\{b\in A:b < w_a\}|} J_A \gamma_{w_1}\cdots \gamma_{w_{a-1}} \gamma_{A\setminus\{w_a\}} \gamma_{w_{a+1}} \cdots \gamma_{w_L}.
\end{align}
Similarly the classical solution to \eqref{eq:dsx} $x_j(s) = e^{s(2\nabla G \cdot \nabla)} x_j \Big|_{x=\sigma}$ has a generator that acts as
\begin{align}
    (2\nabla G \cdot \nabla) x_{w_1}\cdots x_{w_L} = 2\sum_{a=1}^L \sum_{A \ni w_a} J_A x_{w_1} \cdots x_{w_{a-1}} x_{A\setminus\{w_a\}} x_{w_{a+1}} \cdots x_{w_L},
\end{align}
which is identical to the action of $-[H,\cdot]$ in \eqref{eq:hw} up to the absent sign $(-1)^{|\{b\in A:b < w_a\}|}$ that originates from reordering fermions. 
Repeated application of the generators therefore produces corresponding terms in the Taylor expansions of the two evolutions. To compare their moments, we average over the Gaussian disorder, where terms in which a coupling appears an odd number of times vanish, while the remaining Gaussian moments are nonnegative. Bounding the fermionic signs by the triangle inequality, together with the fermionic trace and the average over Ising configurations, gives \eqref{eq:jcomp}. The comparison thus retains the disorder cancellations before bounding the contributions from fermionic signs.

It remains to bound the classical Taylor coefficients. Appendix \ref{app:classical} obtains these bounds using Gaussian concentration and estimates for derivatives of polynomials. Through \eqref{eq:jcomp}, they give high-probability bounds on $D\gamma(s)$. An analogous comparison involving $\partial_s D x(s;\sigma)$ controls the imaginary-time derivative of the fermionic Jacobian and supplies the estimate needed in \eqref{eq:ccomp}. Together, these bounds establish the gap of the unfiltered pseudo-Lindbladian and transfer it to the filtered sampler, which completes the fast-mixing argument.

\subsection{Discussion}
Our work is part of a long line of efforts to show fast mixing of quantum Gibbs samplers at high temperature. In the qubit setting, prior proofs based on Lieb-Robinson bounds~\cite{rouze2025efficient,rouze2026optimal,bakshi2026dobrushin,bergamaschi2026fast} and cluster expansions~\cite{bergamaschi2026fastall} showed fast mixing at large constant temperature. However, these results hold only at completely separable temperatures (hotter than the death of entanglement transition) and admit polynomial-time classical algorithms that prepare mixtures of product states, even in the presence of all-to-all interactions~\cite{putterman2026quantum}. In the fermionic setting, fast mixing has only been proven for sufficiently weak interactions (via Lieb-Robinson bounds)~\cite{h1dx-ps5p, vsmid2025polynomial, vsmid2025rapid}, whereas relevant applications such as Fermi-Hubbard and high-temperature superconductivity lie in the strongly coupled regime that is challenging for classical algorithms~\cite{arovas2022hubbard,qin2022hubbard}. Proofs of fast mixing in the strongly interacting (or, on qubits, far some separable) regime are thus especially interesting.

Our proof of fast mixing departs from both the prior Lieb-Robinson and cluster expansion approaches, and instead more closely resembles classical arguments based on Bakry-\'Emery~\cite{bakry1985diffusions,boudou2006spectral}. Combined with renormalization, these classical results also show log Sobolev inequalities for high-temperature spin systems, including the Sherrington-Kirkpatrick model~\cite{bauerschmidt2019very}. (See~\cite{carlen2017gradient,datta2020relating,carlen2020non} for a similar quantum argument.) Moreover, our proof strategy is fairly generic: the only SYK-specific step is finding a suitable classical Jacobian to bound the quantum Jacobian. We thus expect that our proof strategy straightforwardly extends to other relevant strongly interaction fermionic systems by choosing a different classical Jacobian. It would be interesting if rapid mixing $t_{\rm mix}(\epsilon) = O(\log N + \log 1/\epsilon)$ could also be shown, although our comparison step to dynamics with KMS detailed balance is currently too loose.

Although recent works give efficient classical algorithms for computing SYK thermal expectations of local observables~\cite{zlokapa2026syk,zlokapa2026rigorous,schmidhuber2026polynomial}, our result shows that the Gibbs state can be quantumly realized up to arbitrary inverse polynomial trace distance. Due to the large circuit complexity of the SYK Gibbs state at constant temperature~\cite{anschuetz2025strongly}, one may hope that questions about the dynamics caused by exciting the thermal state may lead to quantum advantage; if so, it will be necessary to prepare the Gibbs state faithfully, as guaranteed by our algorithm. Although any $\beta > 0$ suffices to satisfy the circuit complexity lower bounds~\cite{zlokapa2026rigorous} and thus may suffice for quantum advantage with SYK (if it exists), it remains interesting to show fast mixing at all $\beta$.

Conversely, it is also unclear if our proof of fast mixing gives rise to a new route for dequantization. The circuit complexity lower bounds of~\cite{anschuetz2025strongly} concern the size (not depth) of the quantum circuit; although Gaussian states fail to represent SYK Gibbs states, other efficient classical representations may exist. Nonetheless, it does not seem obvious how to dequantize our proof technique. (Notably, our result applies in a far more quantum regime than the prior high-temperature results for separable Gibbs states~\cite{putterman2026quantum}.) It would be interesting if a deeper relationship exists between the quantum Jacobian and the classical Jacobian used to bound it, with the aim of developing of new classical algorithms or understanding more fundamental physical similarities between SYK and classical models~\cite{facoetti2019classical}.

\section{Preliminaries}
Let $N \geq 4$ be even, and let $\cH$ be a Hilbert space of dimension $d = 2^{N/2}$.
We consider Hermitian Majorana operators $\gamma_1, \ldots , \gamma_N$ on $\cH$ satisfying $\{\gamma_i, \gamma_j\} = 2\delta_{ij} I$. For a subset $S = \{s_1 < \cdots < s_k\} \subseteq [N]$, we denote the ordered Majorana string by
\begin{align}
    \gamma_S = \gamma_{s_1}\cdots \gamma_{s_k}, \quad \gamma_{\emptyset} = I.
\end{align}
Note that due to anticommutation,
\begin{align}
    \gamma_S^{\dagger} = (-1)^{k(k-1)/2} \gamma_S.
\end{align}
Each $\gamma_S$ is unitary; when $|S|=4$ it is also Hermitian. 
Hence, we can write the SYK Hamiltonian as
\begin{align}
H=\sum_{\substack{A\subseteq[N],\\ |A|=4}}J_A\gamma_A,
\qquad
J_A\overset{\mathrm{iid}}{\sim}
\mathcal N\left(0,\frac{6J^2}{N^3}\right),
\end{align}
where $J > 0$ is a constant specifying the coupling scale.

To study the dynamics, we regard the space of operators on $\cH$, denoted by $\mathcal{B}(\cH)$, as a Hilbert space with the inner product 
\begin{align}
    \langle X,Y\rangle_{\text{HS}} = \Tr(X^\dagger Y).
\end{align}
All traces are unnormalized, so $\Tr(I) = d$.
By orthogonality $\Tr(\gamma_S^{\dagger} \gamma_T) = d \delta_{S,T}$, an orthonormal basis of $\cB(\cH)$ is
\begin{align}
    \{d^{-1/2}\gamma_S : S \subseteq [N]\}.
\end{align}
We write $\|X\|$ for the standard operator norm and 
\begin{align}
    \|X\|_p = (\Tr|X|^p)^{1/p}, \quad |X| = (X^\dagger X)^{1/2},
\end{align}
for the Schatten $p$-norm, $1 \leq p < \infty$.
For a linear map $\cM$ acting on operator space, its Hilbert-Schmidt operator norm is 
\begin{align}
    \|\cM\|_{2\to 2} = \underset{X \neq 0}{\sup} \frac{\|\cM(X)\|_2}{\|X\|_2}.
\end{align}
We also use the commutator map 
\begin{align}
    \text{ad}_H(X) = [H,X] = HX - XH.
\end{align}
Since $H$ is Hermitian, $\text{ad}_H$ is self-adjoint on the Hilbert-Schmidt space.

Fermion parity distinguishes states with even and odd fermion number. 
We use this distinction to describe which states an operator can connect. 
With our sign convention, the physical parity operator is 
\begin{align}
    P_f = i^{N/2} \gamma_1 \cdots \gamma_N.
\end{align}
It is Hermitian and squares to identity, so it has eigenvalues $+1$ and $-1$.
An operator is even if it preserves each sector and odd if it exchanges them. 
Equivalently, the parity map 
\begin{align}
    \cP(X) = P_f X P_f
\end{align}
satisfies $\cP(X) = X$ for even operators and $\cP(X) = -X$ for odd operators.
Each Majorana is odd, so $\cP(\gamma_S) = (-1)^{|S|}\gamma_S$. Note that any operator decomposes into orthogonal even and odd parts, i.e., $X = (X+\cP(X))/2 + (X-\cP(X))/2$.

\subsection{Fermionic Jacobian}
The main technical step of our proof uses the Jacobian associated with imaginary time evolution. Recall that the derivative on the Clifford algebra is
\begin{align}
    \partial_i(X) = \frac{1}{2}\lr{\gamma_i X - \cP(X)\gamma_i}, \qquad \partial_i(XY) = (\partial_i X)Y + \cP(X) \partial_i Y.
\end{align}
It removes $\gamma_i$ from an ordered Majorana string when $i$ is present and otherwise returns zero. Because $H$ is even, an imaginary time-evolved fermion $\gamma_j(s) = e^{-sH} \gamma_j e^{sH}$ is odd. Consequently, its Jacobian is
\begin{align}
    (D\gamma)_{ji}(s) = \partial_i \gamma_j(s) = \frac{1}{2}\{\gamma_i, \gamma_j(s)\}, \qquad D\gamma(0) = I_N \otimes I_d.
\end{align}
The gap argument will use that $D\gamma(\beta/2)$ is close to the identity. The intuition for this, if the Hamiltonian were defined on classical states, is that perturbing the coordinate at $i$ will not strongly affect the imaginary-time-evolved coordinate $j$. To show $D\gamma(\beta/2)$ is close to the identity, we will later apply this classical intuition and directly bound $D\gamma$ by the Jacobian of a classical Hamiltonian.

We will also use $\partial_s (D\gamma)_{ij}(s)$. Define the antisymmetric tensor $J_{ijkl}$ by $J_{ijkl} = J_{\{i,j,k,l\}}$ for $i < j < k < l$ and $J_{ijkl}=0$ when indices repeat. For operators $(A_i)_{i=1}^N$ and $(B_i)_{i=1}^N$, define
\begin{align}\label{eq:mdef}
    \sM(A,B)_{ij}(X) = \frac{1}{3} \sum_{k,l} J_{ijkl} (A_k A_l X - A_k X B_l + X B_k B_l).
\end{align}
We apply $\sM$ to operators $X_1, \dots, X_N$ as
\begin{align}
    \lr{\sM(A,B)(X_1,\dots,X_N)}_i = \sum_j \sM(A,B)_{ij}(X_j),
\end{align}
and use its induced norm for $\norm{(X_1,\dots,X_N)}_2^2 = \sum_i \norm{X_i}_2^2$.
This operator is defined such that
\begin{align}
    \partial_s (D\gamma)_{ij}(s) &= \frac{1}{3}\sum_{a,k,l} J_{iakl}\lr{(D\gamma)_{aj}(s)\gamma_k(s)\gamma_l(s) + \gamma_k(s)(D\gamma)_{aj}(s)\gamma_l(s) + \gamma_k(s)\gamma_l(s)(D\gamma)_{aj}(s)}\label{eq:dmg}\\
    &= \sum_a \sM(\gamma(s), -\gamma(s))_{ia}\lr{(D\gamma)_{aj}(s)}.\label{eq:mprop}
\end{align}

\subsection{Gibbs sampling}\label{subsec:gibbs_sampling}
For inverse temperature $\beta \geq 0$, the Gibbs state of the Hamiltonian $H$ is 
\begin{align}
\rho_{\beta} = \frac{e^{-\beta H}}{\Tr(e^{-\beta H})}.
\end{align}
We call $Z_\beta = \Tr(e^{-\beta H})$ the partition function of $H$.
This state describes thermal equilibrium of the Hamiltonian at inverse temperature $\beta$. 
We write $\rho = \rho_\beta$ when the temperature is fixed. We will also use notation
\begin{align}
    \wt{\gamma}_i=\rho^{1/4}\gamma_i \rho^{-1/4}, \quad \wt{\gamma}_i^\dagger=\rho^{-1/4}\gamma_i \rho^{1/4}.
\end{align}
We prepare $\rho$ through continuous-time evolution 
\begin{align}
    \frac{d\rho_\tau}{d\tau} = \cL(\rho_\tau), \quad \rho_\tau = e^{\tau \cL} (\rho_0),
\end{align}
where $\rho_0$ is the initial density matrix. 
Specifically, we use a KMS-detailed-balanced Gibbs sampler from the construction in \cite{ding2025efficient}, with the Majoranas $\gamma_1, \ldots, \gamma_N$ as coupling operators.
The corresponding Lindblad generator is
\begin{align}
\label{eq:lindblad}
\cL(X) = -i[B, X] + \sum_{i=1}^N \left(A_i X A_i^{\dagger} - \frac{1}{2}\{A_i^{\dagger} A_i, X\}\right).
\end{align}
The jump operators $A_i$ describe dissipative transitions, while the Hermitian operator $B$ supplies the coherent correction required by the construction. 
This Lindblad form makes $e^{\tau \cL}$ completely positive and trace preserving. We use the normalization in \eqref{eq:lindblad}, with a sum over the $N$ dissipative terms and no prefactor of $1/N$. (All spectral gaps and mixing times in this paper refer to this choice of normalization. Replacing $\cL$ by $\cL/N$ would divide the spectral gap by $N$ and multiply the mixing time by $N$.) 
Its jumps are 
\begin{align}
    A_i = e^{-\beta H/4} \underset{Y_i}{\underbrace{\left[\exp\left(-\frac{\beta^2}{8} \operatorname{ad}_H^2\right) (\gamma_i)\right]}}e^{\beta H/4}.
\end{align}
Here $\operatorname{ad}_H^2(X) = [H, [H,X]]$, and the exponential is defined by spectral calculus on Hilbert-Schmidt space.
The operator $Y_i$ is a Hermitian, Gaussian-filtered Majorana.
A matrix element that changes the energy by $w$ is multiplied in $A_i$ by 
\begin{align}
    e^{-\beta w/4 - \beta^2 w^2/8}.
\end{align}
The Gaussian factor suppresses transitions involving large energy changes, whether they increase or decrease the energy.
The factor $e^{-\beta w/4}$ gives greater weight to a transition that lowers the energy than to the reverse transition, which encourages relaxation toward lower energies.
For $\beta >0$, the jumps also admit the real-time integral representation
\begin{align}
    A_i=\int_{\mathbb{R}} f_\beta(s) e^{i s H} \gamma_i e^{-i s H} d s,\quad f_\beta(s)=\frac{\sqrt{2/\pi}}{\beta} \exp\left( \frac{1}{8} +\frac{is}{\beta} - \frac{2s^2}{\beta^2}\right).
\end{align}
At $\beta = 0$, we recover $A_i = Y_i = \gamma_i$.

The \cite{ding2025efficient} construction also contains the Hermitian correction term
\begin{align}
    B = \frac{i}{2} \tanh{\left(\frac{\beta}{4} \text{ad}_H\right)} \left(\sum_{i=1}^N A_i^\dagger A_i\right),
\end{align}
so that the complete generator satisfies KMS detailed balance, the quantum analog of reversibility at thermal equilibrium. 
Precisely, the adjoint generator $\cL^\dagger$ is self-adjoint in the Gibbs-weighted inner product defined below. 
This condition ensures that the Gibbs state is stationary under the dynamics.
At $\beta = 0$, both the filtering and the temperature bias disappear, so $A_i = \gamma_i, B=0$, and $\rho = I/d$. 

The relevant notion of reversibility is Kubo-Martin-Schwinger (KMS) detailed balance. 
It expresses symmetry of the evolution in the Gibbs-weighted inner product 
\begin{align}
    \langle X, Y\rangle_{\rho, \text{KMS}} = \Tr(\rho^{1/2}X^\dagger \rho^{1/2}Y).
\end{align}
The associated norm is 
\begin{align}
    \|O\|_{\rho, \text{KMS}}^2 = \langle O, O\rangle_{\rho, \text{KMS}}.
\end{align}
In our setting, the Hilbert-Schmidt adjoint $\cL^\dagger$, acting on observables, satisfies
\begin{align}
    \langle X, \cL^\dagger (Y)\rangle_{\rho, \text{KMS}} = \langle \cL^{\dagger}(X), Y\rangle_{\rho, \text{KMS}}.
\end{align}
Together with trace preservation, this implies $\cL(\rho) = 0$, and once the system reaches the Gibbs state, the evolution leaves it unchanged. 
The remaining question is then how quickly an arbitrary initial state approaches it. 

To study this rate, we express the generator in coordinates adapted to the Gibbs state. 
The resulting operator, called the discriminant, is 
\begin{align}\label{eq:gap_bound}
    \cD(X) = -\rho^{-1/4} \cL(\rho^{1/4}X\rho^{1/4}) \rho^{-1/4}.
\end{align}
The KMS detailed balance condition makes $\cD$ self-adjoint and nonnegative on the Hilbert-Schmidt space. 
It is similar to $-\cL$, so its eigenvalues are the negatives of those of $\cL$, and its zero eigenvector $\sqrt{\rho}$ represents equilibrium, $\cD(\sqrt{\rho}) = 0$. 
We say the discriminant $\cD$ has spectral gap at least $\lambda > 0$ if 
\begin{align}
    \langle X, \cD(X) \rangle_{\text{HS}} \geq \lambda \|X\|_2^2
\end{align}
for all $X \in \cB(\cH)$ satisfying $\langle \sqrt{\rho}, X\rangle_{\text{HS}} = 0$.
The expression on the left is the Dirichlet form of the generator. 
For $X = \rho^{1/4} O \rho^{1/4}$, it equals one half of the initial rate of decrease of the KMS variance under the observable evolution $e^{t\cL^\dagger}(O)$.
The inequality therefore bounds how quickly observables relax toward their Gibbs expectations. 
The two formulations are related by the transformation $X \to \rho^{1/4} X \rho^{1/4}$, which satisfies 
\begin{align}
    \|\rho^{1/4} X \rho^{1/4}\|_2^2 = \|X\|_{\rho, \text{KMS}}^2, \quad \langle \sqrt{\rho}, \rho^{1/4} X \rho^{1/4}\rangle_{\text{HS}} = \Tr(\rho X).
\end{align}
Thus subtracting the Gibbs expectation in the KMS formulation corresponds to removing the component along $\sqrt{\rho}$ in Hilbert-Schmidt space.

A positive gap bound guarantees that $\rho$ is the unique stationary state and controls convergence to it from any density matrix. 
The gap lower bound $\lambda$ controls the convergence rate, as a larger gap gives a faster convergence rate. 
We quantify this dependence through the trace-norm mixing time defined below.

We also note that the generator preserves operator parity because each $A_i$ is odd and $B$ is even.

We measure the convergence time by how long it takes every initial state to become close to $\rho$ in trace norm.
For $\epsilon \in (0,1)$, define the $\epsilon$-mixing time
\begin{align}
    t_{\text{mix}}(\epsilon) = \inf\left\{\tau \geq 0: \sup_{\rho_0} \|e^{\tau \cL}(\rho_0) - \rho\|_1 \leq \epsilon\right\},
\end{align}
where the supremum is over all initial density matrices $\rho_0$. 
A gap lower bound $\lambda$ yields 
\begin{align}
    \|e^{\tau \cL}(\rho_0) - \rho\|_1 \leq e^{-\lambda \tau}\sqrt{d} e^{\beta \|H\|},
\end{align}
and therefore 
\begin{align}
    t_{\text{mix}}(\epsilon) \leq \frac{1}{\lambda} \left(\frac{1}{2}\log d + \beta \|H\| + \log \frac{1}{\epsilon}\right).
\end{align}
Lastly, since $d = 2^{N/2}$, a constant gap and a bound $\|H\| = O(JN)$ give mixing time $O(N + \log(1/\epsilon))$ when $\beta J$ is bounded.

\section{Spectral gap of pseudo-Lindbladians}\label{sec:pseudo_lindblad}
To bound the spectral gap of the Gibbs sampler, our first step is to relate this gap to that of a positive auxiliary operator built directly from the jump operators $A_i$.
We will show that a lower bound on its gap gives a lower bound on the gap of the discriminant $\cD$, and therefore of the Lindblad generator $\cL$.

The Gaussian-filtered Majorana $Y_i = \exp\left(-\frac{\beta^2}{8} \operatorname{ad}_H^2\right) (\gamma_i)$ is Hermitian because the filter is real and even in the energy difference. 
Since $A_i = \rho^{1/4} Y_i \rho^{-1/4}$, we have the identity
\begin{align}
    A_i \sqrt{\rho} = \rho^{1/4} Y_i \rho^{1/4} = \sqrt{\rho} A_i^\dagger.
\end{align}
For an arbitrary operator $X\in \cB(\cH)$, consider the nonnegative quantity, which we call the quadratic form:
\begin{align}
    \frac{1}{4} \sum_{i=1}^N \|A_i X - XA_i^{\dagger}\|_2^2.
\end{align}
Each summand measures failure of one relation $A_i X = XA_i^{\dagger}$, and the sum vanishes precisely when all these relations hold; in particular, it vanishes for $X = \sqrt{\rho}$.
This sum can be written as 
\begin{align}
    \langle X, \cK_f(X)\rangle_{\text{HS}}
    &=\frac{1}{4} \sum_{i=1}^N \|A_i X - XA_i^{\dagger}\|_2^2 \quad\text{for}\quad \cK_f(X)
    =\frac{1}{4} \sum_{i=1}^N
      \left[A_i^\dagger(A_iX-XA_i^\dagger)-(A_iX-XA_i^\dagger)A_i\right].
\end{align}
We call $\cK_f$ the \emph{filtered pseudo-Lindbladian}.
By construction, $\cK_f$ is self-adjoint and nonnegative, with $\cK_f(\sqrt{\rho}) = 0$. It is related to the pseudo-Lindbladian $\cK^{\rm BC}_Q(X) = [Q,X]\rho^{1/2}Q^\dagger \rho^{-1/2} - \rho^{-1/2}Q^\dagger \rho^{1/2}[Q,X]$ introduced by~\cite{bergamaschi2026fast} by the similarity transform $V(X) = \rho^{1/4} X \rho^{1/4}$:
\begin{align}
    \cK_f = -\frac{1}{4} V \lr{\sum_i \cK^{\rm BC}_{Y_i}} V^{-1}.
\end{align}
We further introduce the \emph{unfiltered pseudo-Lindbladian} $\cK$, obtained by removing the Gaussian filter from each jump while retaining its Gibbs factors. (It is related to $\cK^{\rm BC}_{\gamma_i}$ instead of $\cK^{\rm BC}_{Y_i}$ under the same similarity transform as above.) Equivalently, $\cK$ is the positive self-adjoint operator specified by 
\begin{align}
    \langle X, \cK(X)\rangle_{\text{HS}} = \frac{1}{4}\sum_{i=1}^N \| \wt{\gamma}_i  X - X  \wt{\gamma}_i^\dagger\|_2^2.
\end{align}
This gives 
\begin{align}
    \cK(X) = \frac{1}{4}\sum_{i=1}^N
    \left[\wt{\gamma}_i^\dagger(\wt{\gamma}_iX-X\wt{\gamma}_i^\dagger)
    -(\wt{\gamma}_iX-X\wt{\gamma}_i^\dagger)\wt{\gamma}_i\right].
\end{align}
The same argument used for $\cK_f$ also gives $\cK(\sqrt{\rho}) = 0$, since the unfiltered jumps also satisfy $\wt{\gamma}_i \sqrt{\rho} = \sqrt{\rho} \wt{\gamma_i}^\dagger$ for every $i$.

At $\beta = 0$, the Gaussian filter and Gibbs conjugation both become the identity.
Therefore, $\cK_f = \cK$ and 
\begin{align}
    \langle X, \cK(X)\rangle_{\text{HS}} = \frac{1}{4} \sum_{i=1}^N \|[\gamma_i, X]\|_2^2,
\end{align}
which measures the failure of $X$ to commute with the Majoranas. 
At general positive inverse temperature, the same interpretation holds in the KMS product after the change of variables introduced in \Cref{subsec:gibbs_sampling}.
Specifically 
\begin{align}
    \langle \rho^{1/4}X\rho^{1/4}, \cK\left(\rho^{1/4}X\rho^{1/4}\right)\rangle_{\text{HS}} = \frac{1}{4} \sum_{i=1}^N \|[\gamma_i, X]\|_{\rho, \text{KMS}}^2,
\end{align}
and 
\begin{align}
    \langle \rho^{1/4}X\rho^{1/4}, \cK_f\left(\rho^{1/4}X\rho^{1/4}\right)\rangle_{\text{HS}} = \frac{1}{4} \sum_{i=1}^N \|[Y_i, X]\|_{\rho, \text{KMS}}^2.
\end{align}
Thus comparing the squared commutators with $Y_i$ and $\gamma_i$ in the KMS norm is equivalent to comparing the quadratic forms of $\cK_f$ and $\cK$.

The unfiltered $\cK$ is simpler to analyze because the jumps still square to the identity and anticommute for distinct Majoranas. These algebraic identities reduce the spectral gap argument for $\cK$ to a bound on $D\gamma(\beta/2)$ in \Cref{sec:majorana}, while \eqref{eq:dmg} supplies the estimates needed to compare $\cK$ and $\cK_f$.

\subsection{Filtered pseudo-Lindbladian}
We first prove the remaining comparison: a gap bound for $\cK_f$ implies a gap bound for the Gibbs sampler.
The main observation is that the discriminant $\cD$ is twice a positive weighted average of unitary conjugates of $\cK_f$.
We establish this representation in two steps, then use it to compare the gaps.
Throughout, $Y_i$ and $A_i$ are the Gaussian-filtered operators as defined in \Cref{subsec:gibbs_sampling}:
\begin{align}
Y_i = \exp\left(-\frac{\beta^2}{8} \operatorname{ad}_H^2\right) (\gamma_i), \quad A_i = e^{-\beta H/4} Y_i e^{\beta H/4}.
\end{align}

\begin{lemma}\label{lem:D_average_rep}
    For every $\beta > 0$, the discriminant $\cD$ can be written as
    \begin{align}\label{eq:averaging_identy}
        \cD = 2 \int_{\R} \frac{2}{\beta} \sech\left(\frac{2\pi s}{\beta}\right) e^{is \text{ad}_H} \cK_f e^{-is \text{ad}_H} ds,
    \end{align}
    where $e^{is \operatorname{ad}_H} (X) = e^{isH} X e^{-isH}$.
\end{lemma}
\begin{proof}
    We may rewrite \(A_i=e^{-\beta H/4}Y_i e^{\beta H/4}\) and
    \(A_i^\dagger=e^{\beta H/4}Y_i e^{-\beta H/4}\).
    Substituting the Lindblad generator into the definition of $\cD$, and cancelling the scalar normalization factors from $\rho^{\pm 1/4}$, we obtain 
    \begin{align}\label{eq:discriminant_plugin}
        \cD(X) = e^{\beta H/4} \left(\frac{1}{2} \sum_i A_i^\dagger A_i + iB\right) e^{-\beta H/4} X + X e^{-\beta H/4} \left(\frac{1}{2}\sum_i A_i^\dagger A_i - iB\right)e^{\beta H/4} - \sum_i Y_i X Y_i.
    \end{align}
    In particular, the jump term becomes $Y_iXY_i$ because 
    \begin{align}
        e^{\beta H/4} A_i e^{-\beta H/4} = e^{-\beta H/4} A_i^{\dagger} e^{\beta H/4} = Y_i.
    \end{align}
    The choice of $B = \frac{i}{2} \tanh{\left(\frac{\beta}{4} \text{ad}_H\right)} \left(\sum_{i=1}^N A_i^\dagger A_i\right)$ simplifies the two remaining terms.
    Indeed, 
    \begin{align}
         e^{\pm (\beta/4) \text{ad}_H} \left(\frac{1}{2} \sum_i A_i^{\dagger} A_i \pm iB\right) &= \frac{1}{2} e^{\pm (\beta/4) \text{ad}_H} \left(I \mp  \tanh\left(\frac{\beta}{4} \text{ad}_H\right)\right) \left(\sum_i A_i^\dagger A_i\right) \\
         &= \frac{1}{2} \sech\left(\frac{\beta}{4} \text{ad}_H\right) \left(\sum_i A_i^\dagger A_i\right).
    \end{align}
    The last equality follows by the scalar identity $e^{\pm x} (1\mp \tanh{x}) = \sech(x)$.
    Therefore, \eqref{eq:discriminant_plugin} becomes 
    \begin{align}\label{eq:discriminant_decomp}
        \cD(X) = \frac{1}{2} \sum_i \left[\sech\left(\frac{\beta}{4} \text{ad}_H\right) (A_i^\dagger A_i)\right]X + \frac{1}{2} \sum_i X\left[\sech\left(\frac{\beta}{4} \text{ad}_H\right) (A_i^\dagger A_i)\right]- \sum_i Y_i X Y_i.
    \end{align}
    We now verify \eqref{eq:averaging_identy} by averaging the corresponding expansion of $\cK_f$:
    \begin{align}\label{eq:pseudo_lindblad_filtered}
        \cK_f(X) = \frac{1}{4}\sum_i (A_i^\dagger A_i X + XA_i^\dagger A_i - A_i^\dagger X A_i^\dagger - A_i X A_i).
    \end{align}
    Since $e^{is \text{ad}_H}(X) = e^{isH}X e^{-isH}$, conjugating $\cK_f$ by $e^{is \text{ad}_H}$ conjugates each matrix multiplying $X$ by $e^{isH}$.
    The effect of averaging these conjugations is determined by the scalar Fourier identity 
    \begin{align}\label{eq:scalar_fourier}
        \int_{\R} \frac{2}{\beta} \sech\left(\frac{2\pi s}{\beta}\right) e^{is\omega} ds = \sech\left(\frac{\beta \omega }{4}\right).
    \end{align}
    To recover the first two terms in \eqref{eq:discriminant_decomp}, spectral calculus and \eqref{eq:scalar_fourier} give 
    \begin{align}
        \int_{\R} \frac{2}{\beta} \sech\left(\frac{2\pi s}{\beta}\right) e^{is H} A_i^\dagger A_i e^{-isH} ds = \sech\left(\frac{\beta}{4} \text{ad}_H\right)(A_i^\dagger A_i).
    \end{align}
    Thus, their averages produce the left and right multiplication terms that appear in \eqref{eq:discriminant_decomp}.
    To treat the remaining two cross terms, choose an orthonormal energy eigenbasis $H \ket{m} = E_m \ket{m}$. 
    The definitions of $A_i$ and $Y_i$ imply 
    \begin{align}
        (A_i)_{mp} = e^{-\beta(E_m - E_p)/4} (Y_i)_{mp}, \quad (A_i^\dagger)_{mp} = e^{\beta(E_m - E_p)/4} (Y_i)_{mp}.
    \end{align}
    For the map $X \to A_i X A_i + A_i^\dagger X A_i^\dagger$, the coefficient taking $\ket{p}\bra{q}$ to $\ket{m}\bra{n}$ is therefore 
    \begin{align}
        (A_i)_{mp} (A_i)_{qn} + (A_i^\dagger)_{mp} (A_i^\dagger)_{qn} = 2 \cosh\left(\frac{\beta \omega }{4}\right) (Y_i)_{mp} (Y_i)_{qn},
    \end{align}
    where $\omega = (E_m - E_p) + (E_q - E_n)$.
    Conjugating this map by $e^{is \text{ad}_H}$ multiplies the coefficient by 
    \begin{align}
        e^{is ((E_m - E_n) - (E_p - E_q))} = e^{is\omega}.
    \end{align}
    Combining with \eqref{eq:scalar_fourier}, the scalar factor after averaging is then 
    \begin{align}
        2\cosh\left(\frac{\beta \omega}{4}\right) \int_{\R} \frac{2}{\beta} \sech\left(\frac{2\pi s}{\beta}\right)e^{is\omega} ds = 2\cosh\left(\frac{\beta\omega}{4} \right)\sech\left(\frac{\beta \omega}{4}\right) = 2.
    \end{align}
    We then get that the average of $X \to A_i X A_i + A_i^\dagger X A_i^\dagger$ is $X \to 2Y_i X Y_i$.
    Combining these calculations with the signs and the factor of $1/4$ in \eqref{eq:pseudo_lindblad_filtered}, we obtain the desired expression in \eqref{eq:averaging_identy}.
\end{proof}

\begin{lemma}[From the filtered pseudo-Lindbladian to the Gibbs sampler]\label{lem:filtered_to_lind_gap}
    Assume that there exists $a > 0$ such that 
    \begin{align}
        \langle X, \cK_f(X)\rangle_{\text{HS}} \geq a \|X\|_2^2
    \end{align}
    for every $X \in \cB(\cH)$ satisfying $\langle \sqrt{\rho}, X \rangle_{\text{HS}} = 0.$
    Then 
    \begin{align}
        \text{gap}(\cL) \geq 2a.
    \end{align}
\end{lemma}
\begin{proof}
    We first consider $\beta > 0$, and let $X$ be orthogonal to $\sqrt{\rho}$.
    To apply the assumed bound for $\cK_f$ inside the average in \Cref{lem:D_average_rep}, we must check that the operators $e^{-isH} X e^{isH}$ satisfy the same orthogonality condition.
    Since $H$ commutes with $\sqrt{\rho}$ and $e^{isH}$ is unitary, 
    \begin{align}
        \langle \sqrt{\rho}, e^{-isH}X e^{isH}\rangle_{\text{HS}} = \langle \sqrt{\rho}, X\rangle_{\text{HS}} = 0, \quad \|e^{-isH}Xe^{isH}\|_2 = \|X\|_2.
    \end{align}
    The assumption of $\langle X, \cK_f(X)\rangle_{\text{HS}} \geq a \|X\|_2^2$ then gives the same lower bound for every $s \in \R$: 
    \begin{align}
        \langle e^{-isH}Xe^{isH}, \cK_f(e^{-isH}Xe^{isH}) \rangle_{\text{HS}} \geq a \|X\|_2^2.
    \end{align}
    Inserting this inequality into the averaging identity from \Cref{lem:D_average_rep} gives us 
    \begin{align}
        \langle X, \cD(X) \rangle_{\text{HS}} &= 2 \int_{\R} \frac{2}{\beta} \sech\left(\frac{2\pi s}{\beta}\right) \langle e^{-isH} X e^{isH}, \cK_f(e^{-isH}X e^{isH})\rangle_{\text{HS}} ds\\
        &\geq 2a \|X\|_2^2 \int_{\R} \frac{2}{\beta} \sech\left(\frac{2\pi s}{\beta}\right) ds\\
        &= 2a \|X\|_2^2, \label{eq:bound_quad_term}
    \end{align}
    The weight here is nonnegative and integrates to one.
    Moreover, the conjugations in the integrand preserve the Hilbert-Schmidt norm and fix $\sqrt{\rho}$. 
    
    For $\beta = 0$, the Gaussian filter becomes the identity, so
    \begin{align}
        A_i = \gamma_i, \quad B = 0, \quad \rho = I/d,
    \end{align}
    which gives us 
    \begin{align}
        \cD(X) = \frac{1}{2} \sum_i [A_i,[A_i,X]] = 2\cK_f(X),
    \end{align}
    so the same lower bound in \eqref{eq:bound_quad_term} holds.

    Since $\cD$ is self-adjoint and annihilates $\sqrt{\rho}$, it admits an orthonormal eigenbasis containing $\sqrt{\rho}$.
    Every other eigenvector $X$, with eigenvalue $\lambda$, is orthogonal to $\sqrt{\rho}$. 
    Applying the bound in \eqref{eq:bound_quad_term} gives us 
    \begin{align}
        \lambda \|X\|_2^2 = \langle X, \cD(X)\rangle_{\text{HS}} \geq 2a \|X\|_2^2,
    \end{align}
    and therefore $\lambda \geq 2a$.
    Lastly, the definition of $\cD$ implies 
    \begin{align}
        \cL(\rho^{1/4} X \rho^{1/4}) = -\lambda \rho^{1/4} X \rho^{1/4}.
    \end{align}
    The transformation $X\to \rho^{1/4} X \rho^{1/4}$ is invertible, so $\cD$ and $-\cL$ have the same eigenvalues. 
    Therefore, 
    \begin{align}
        \operatorname{gap}(\cL) = \operatorname{gap}(\cD) \geq 2a.
    \end{align}
\end{proof}

\subsection{Unfiltered pseudo-Lindbladian}
The above \Cref{lem:filtered_to_lind_gap} reduces the desired gap bound for the Gibbs sampler to a lower bound for the filtered pseudo-Lindbladian $\cK_f$.
To use the estimates for the unfiltered operator $\cK$ developed in \Cref{sec:majorana}, we must show that filtering preserves a sufficiently large fraction of this bound. 
We do this by controlling how much the filter changes the individual jump differences. 

The filter acts on the jumps before these squared norms are formed.
To track its effect, it is useful to name the individual differences already appearing in the definitions of $\cK$ and $\cK_f$, which we define as 
\begin{align}\label{eq:c_defn}
    c_i(X) = \frac{1}{2} \left(\wt{\gamma}_i X - X \wt{\gamma}_i^\dagger \right), \quad c_i^f(X) = \frac{1}{2} (A_iX - XA_i^\dagger).
\end{align}
Each map above records the failure of one equilibrium relation, and $c_i(X)$ uses the unfiltered jump while $c_i^f(X)$ uses the filtered jump. 
In particular, both vanish at $X = \sqrt{\rho}$.
Their squared norms recover the quantities we wish to compare: 
\begin{align}
    \langle X, \cK(X) \rangle_{\text{HS}} = \sum_i \|c_i(X)\|_2^2, \quad \langle X, \cK_f(X)\rangle_{\text{HS}} = \sum_i \|c_i^f(X)\|_2^2.
\end{align}
Therefore it suffices to bound the difference between the tuples $\left(c_i^f(X)\right)_{i=1}^N$ and $(c_i(X))_{i=1}^N$.

For $\beta >0$, the filtered jump $A_i$ can be written as a weighted average of Hamiltonian evolutions of $\wt{\gamma_i} = e^{-\beta H/4} \gamma_i e^{\beta H/4}$:
\begin{align}
    A_i = \frac{\sqrt{2/\pi}}{\beta}\int_{\R} e^{-2s^2/\beta^2} e^{isH} \wt{\gamma}_i e^{-isH} ds.
\end{align}
Using the definitions of $c_i$ and $c_i^f$, the same identity gives 
\begin{align}
    c_i^f(X) = \frac{\sqrt{2/\pi}}{\beta} \int_{\R} e^{-2s^2/\beta^2} e^{isH} c_i\left(e^{-isH} X e^{isH}\right) e^{-isH} ds.
\end{align}
Thus $c_i^f (X)$ is an average of the conjugated expressions in this integral, whose value at $s=0$ is $c_i(X)$.
To show that filtering changes $c_i(X)$ only slightly, we therefore estimate how much these expressions change as $s$ moves away from zero. 

The jump differences are controlled by $\sM$ \eqref{eq:mdef} applied to tuples $\wt \gamma = \gamma(\beta/4)$ and $\wt \gamma^\dagger = \gamma(-\beta/4)$. The following identity explains how $\mathsf{M}$ controls the rate of change of the jump differences under Hamiltonian evolution relative to their norm. 

\begin{lemma}\label{lem:commutator_identity}
    For SYK Hamiltonian $H$ and the maps $c_i$ and $\mathsf{M}(\wt\gamma,\wt\gamma^\dagger)_{ij}$ as defined in \eqref{eq:c_defn} and \eqref{eq:mdef} respectively, the following identity holds 
    \begin{align}
        c_i([H,X]) - [H, c_i(X)] = \sum_j \mathsf{M}(\wt\gamma,\wt\gamma^\dagger)_{ij}(c_j(X))
    \end{align}
    for every $X \in \cB(\cH)$ and every $i$.
\end{lemma}
\begin{proof}
    The Majorana relations and antisymmetry of the couplings give 
    \begin{align}
        [\gamma_i, H] = \frac{1}{3} \sum_{j,k,l} J_{ijkl}\gamma_j \gamma_k \gamma_l.
    \end{align}
    Since the Gibbs factors commute with $H$, conjugating this identity gives 
    \begin{align}
        [\wt{\gamma}_i, H] = \frac{1}{3} \sum_{j,k,l} J_{ijkl}\wt{\gamma}_j \wt{\gamma}_k \wt{\gamma}_l, \quad [\wt{\gamma}_i^\dagger, H] = \frac{1}{3} \sum_{j,k,l} J_{ijkl}\wt{\gamma}_j^\dagger \wt{\gamma}_k^\dagger \wt{\gamma}_l^\dagger.
    \end{align}
    Directly expanding $c_i$, we obtain 
    \begin{align}
         c_i([H,X]) - [H, c_i(X)] &= \frac{1}{2} \left([\wt{\gamma}_i, H]X - X[\wt{\gamma}_i^\dagger, H]\right)\\
         &= \frac{1}{6}\sum_{j,k,l} J_{ijkl} \left(\wt{\gamma}_j \wt{\gamma}_k \wt{\gamma_l}X - X\wt{\gamma}_j^\dagger \wt{\gamma_k}^\dagger \wt{\gamma}_l^\dagger\right).
    \end{align}
    Note that the difference of the two cubic products can be written in terms of the maps $c_j, c_k, c_l$:
    \begin{align}
        \wt{\gamma}_j \wt{\gamma}_k \wt{\gamma}_lX - X\wt{\gamma}_j^\dagger \wt{\gamma}_k^\dagger \wt{\gamma}_l^\dagger = 2\wt{\gamma}_j \wt{\gamma}_k c_l(X) + 2\wt{\gamma}_j c_k(X)\wt{\gamma}_l^\dagger + 2c_j(X)\wt{\gamma}_k^\dagger \wt{\gamma}_l^\dagger. 
    \end{align}
    To express every contribution in terms of $c_j(X)$, we rename the summation indices in the first two terms:
    \begin{align}
        \sum_{j,k,l} J_{ijkl}\wt{\gamma}_j \wt{\gamma}_k c_l(X) = \sum_{j,k,l} J_{iklj}\wt{\gamma}_k \wt{\gamma}_l c_j(X) = \sum_{j,k,l} J_{ijkl}\wt{\gamma}_k\wt{\gamma}_l c_j(X),
    \end{align}
    and similarly 
    \begin{align}
        \sum_{j,k,l} J_{ijkl} \wt{\gamma}_j c_k(X) \wt{\gamma}_l^\dagger = \sum_{j,k,l} J_{ikjl} \wt{\gamma}_k c_j(X) \wt{\gamma}_l^\dagger.
    \end{align}
    Here antisymmetry gives $J_{iklj} = J_{ijkl}$ for the cyclic permutation and $J_{ikjl} = -J_{ijkl}$ for the exchange of $j$ and $k$.
    The third term already contains $c_j(X)$.
    Combining these three terms then gives us 
    \begin{align}
        c_i([H,X]) - [H, c_i(X)] &= \frac{1}{3} \sum_{j,k,l} J_{ijkl} \left(\wt{\gamma}_k\wt{\gamma}_l c_j(X) - \wt{\gamma}_k c_j(X) \wt{\gamma}_l^\dagger + c_j(X)\wt{\gamma}_k^\dagger \wt{\gamma}_l^\dagger\right) \\
        &= \sum_j \mathsf{M}(\wt\gamma,\wt\gamma^\dagger)_{ij} (c_j(X)).
    \end{align}
\end{proof}
We now have an estimate that controls how Hamiltonian conjugation changes the unfiltered expression $c_i(X)$.
Since Gaussian averaging of these conjugated expressions gives $c_i^f(X)$, we can next bound the difference between $c_i^f(X)$ and $c_i(X)$.
The condition on $\beta J$ ensures that the resulting difference is small. 

\begin{lemma}\label{lem:filter_effect}
    If $0 \leq \beta J \leq \frac{1}{1000}$ and $\|\mathsf{M}(\wt\gamma,\wt\gamma^\dagger)\| \leq 150J$, then for every $X \in \cB(\cH)$,
    \begin{align}
        \left(\sum_i \left\|c_i^f(X) - c_i(X)\right\|_2^2\right)^{1/2} \leq \frac{1}{14} \left(\sum_i \|c_i(X)\|_2^2\right)^{1/2}.
    \end{align}
\end{lemma}
\begin{proof}
    At $\beta = 0$, we have $A_i = \wt{\gamma}_i = \gamma_i$, so $c_i^f = c_i$ and the desired bound directly follows. 
    Assume now $\beta > 0$.
    By the Gaussian Fourier identity 
    \begin{align}
        e^{-\beta^2 w^2 /8} = \frac{\sqrt{2/\pi}}{\beta} \int_{\R} e^{-2s^2/\beta^2} e^{isw} ds,
    \end{align}
    we can write 
    \begin{align}
        Y_i = \frac{\sqrt{2/\pi}}{\beta} \int_{\R} e^{-2s^2/\beta^2} e^{isH} \gamma_i e^{-isH} ds.
    \end{align}
    Since the Gibbs factors commute with Hamiltonian evolution, the same average maps $\wt{\gamma}_i$ to $A_i$.
    Therefore,
    \begin{align}
        c_i^f(X) = \frac{\sqrt{2/\pi}}{\beta} \int_{\R} e^{-2s^2/\beta^2} e^{isH} c_i\left(e^{-isH} X e^{isH}\right) e^{-isH}ds.
    \end{align}
     To estimate how much this average differs from the unfiltered $c_i(X)$, fix $X$ and write 
    \begin{align}
        v_i(s) = e^{isH} c_i(e^{-isH}X e^{isH}) e^{-isH}, \quad v(s) = (v_1(s), \ldots, v_N(s)).
    \end{align}
    Thus
    \begin{align}
        v(0) = (c_1(X), \ldots, c_N(X))
    \end{align}
    is the tuple of unfiltered jump differences, and its weighted average is the tuple of filtered jump differences.
    We first bound how much this tuple changes with $s$.
    Differentiating and using \Cref{lem:commutator_identity}, we obtain 
    \begin{align}
        v'_i(s) &= i e^{isH} \left([H, c_i(e^{-isH} X e^{isH})] - c_i([H, e^{-isH} X e^{isH}])\right) e^{-isH} \\
        &= -i \sum_j e^{isH} \mathsf{M}(\wt\gamma,\wt\gamma^\dagger)_{ij} \left(e^{-isH} v_j(s) e^{isH}\right) e^{-isH}.
    \end{align}
    Conjugating by $e^{isH}$ preserves the Hilbert-Schmidt norm of every component, so the assumed bound on the block operator $\mathsf{M}(\wt\gamma,\wt\gamma^\dagger)$ gives
    \begin{align}
        \|v'(s)\|_2 \leq \|\mathsf{M}(\wt\gamma,\wt\gamma^\dagger)\| \|v(s)\|_2 \leq 150J \|v(s)\|_2.
    \end{align}
    Applying Grönwall's inequality to $v(s)$ and $v(-s)$ for $s \geq 0$ gives 
    \begin{align}
        \|v(s)\|_2 \leq e^{150J|s|} \|v(0)\|_2.
    \end{align}
    Integrating the derivative bound along the interval between $0$ and $s$, we then obtain 
    \begin{align}\label{eq:bound_of_v_diff}
        \|v(s) - v(0)\|_2 \leq 150J \int_{0}^{|s|} e^{150Jr} dr \|v(0)\|_2 = \left(e^{150 J|s|} -1\right) \|v(0)\|_2.
    \end{align}
    We now use this bound to estimate the effect of filtering. 
    The Gaussian integral representation of $c_i^f(X)$ gives 
    \begin{align}
        \left(c_i^f(X)\right)_{i=1}^N = \frac{\sqrt{2/\pi}}{\beta} \int_{\R} e^{-2s^2/\beta^2} v(s) ds.
    \end{align}
    Because the Gaussian weight integrates to one, subtracting the unfiltered tuple gives 
    \begin{align}
        \left(c_i^f(X) - c_i(X)\right)_{i=1}^N = \frac{\sqrt{2/\pi}}{\beta} \int_{\R} e^{-2s^2/\beta^2} (v(s) - v(0)) ds.
    \end{align}
    By the triangle inequality and \eqref{eq:bound_of_v_diff}, we get 
    \begin{align}
        \left(\sum_i \|c_i^f(X) - c_i(X)\|_2^2\right)^{1/2} &\leq \frac{\sqrt{2/\pi}}{\beta} \int_{\R} e^{-2s^2/\beta^2} \|v(s) - v(0)\|_2 ds\\
        &\leq \left(\sum_i \|c_i(X)\|_2^2\right)^{1/2} \frac{\sqrt{2/\pi}}{\beta} \int_{\R} e^{-2s^2/\beta^2} \left(e^{150J|s|}-1\right) ds,
    \end{align}
    where we also used $\|v(0)\|_2^2 = \sum_i \|c_i(X)\|_2^2$ for the last inequality. 
    It remains to bound the scalar integral. 
    Consider the bound
    \begin{align}
        &\frac{\sqrt{2/\pi}}{\beta}\int_{\R} e^{-2s^2/\beta^2} \left(e^{150J|s|} - 1\right) ds = \sqrt{\frac{2}{\pi}} \int_{0}^{\infty} e^{-r^2/2} \left(e^{75\beta Jr} - 1\right) dr\\
        &\leq \sqrt{\frac{2}{\pi}} \int_{0}^{\infty} e^{-r^2/2} \left(e^{3r/40} - 1\right) dr = e^{9/3200} \left(1 + \sqrt{\frac{2}{\pi}} \int_{0}^{3/40} e^{-r^2/2} dr\right) - 1\\
        &\leq \frac{1 + (5/6) (3/40)}{1-9/3200} - 1 = \frac{209}{3191} < \frac{1}{14}.
    \end{align}
    Here, the first equality uses the symmetry of the integrand and the substitution $r = 2s/\beta$.
    The first inequality follows from the assumption $\beta J \leq 1/1000$, which gives $75\beta J \leq 3/40$.
    The next equality follows by completing the square,
    \begin{align}
        -\frac{r^2}{2} + \frac{3r}{40} = -\frac{1}{2} \left(r-\frac{3}{40}\right)^2 + \frac{9}{3200}.
    \end{align}
    Lastly, the second inequality uses $e^{-r^2/2} \leq 1$, $\sqrt{2/\pi} < 5/6$, and $e^x \leq (1-x)^{-1}$ for $0 \leq x < 1$, applied at $x = 9/3200$.
    This completes the proof.
\end{proof}

We can now directly relate $\cK_f$ and $\cK$.
\begin{corollary}[Comparison of pseudo-Lindbladians]\label{cor:pseudo_comparisons}
    If $0\leq\beta J \leq 1/1000$ and $\|\mathsf{M}(\wt\gamma,\wt\gamma^\dagger)\| \leq 150 J$, then for every $X \in \mathcal{B}(\mathcal{H})$, 
\begin{align}
    \langle X, \cK_f(X)\rangle_{\text{HS}} \geq \left(\frac{13}{14}\right)^2 \langle X, \cK(X) \rangle_{\text{HS}}.
\end{align}
Moreover, $\ker{\cK_f} = \ker{\cK}$.
\end{corollary}
\begin{proof}
    From \Cref{lem:filter_effect} and by the reverse triangle inequality 
    \begin{align}
        \left(\sum_i \left\|c_i^f(X)\right\|_2^2\right)^{1/2} &\geq \left(\sum_i \|c_i(X)\|_2^2\right)^{1/2} - \left(\sum_i \|c_i^f(X) - c_i(X)\|_2^2\right)^{1/2}\\ &\geq \frac{13}{14} \left(\sum_i \|c_i(X)\|_2^2\right)^{1/2}.
    \end{align}
    Squaring and using the definitions of $\cK_f$ and $\cK$, we obtain 
    \begin{align}\label{eq:K_Kf_relation}
        \langle X, \cK_f(X)\rangle_{\text{HS}} &= \sum_i \left\|c_i^f(X)\right\|_2^2 \geq \left(\frac{13}{14}\right)^2 \sum_i \|c_i(X)\|_2^2 = \left(\frac{13}{14}\right)^2 \langle X, \cK(X) \rangle_{\text{HS}}.
    \end{align}
    For both of the pseudo-Lindbladians, an operator belongs to one's kernel exactly when all the corresponding maps $c_i$ or $c_i^f$ send it to zero. 
    \eqref{eq:K_Kf_relation} then gives us $\ker \cK_f \subseteq \ker \cK$.
    Conversely, if $X \in \ker \cK$, then $c_i(X) = 0$ for every $i$. 
    \Cref{lem:filter_effect} then gives 
    \begin{align}
        \left(\sum_i \left\|c_i^f(X)\right\|_2^2\right)^{1/2} = 0,
    \end{align}
    so $X \in \ker \cK_f$. 
    This completes the proof.
\end{proof}

\section{Bounding the gap of $\cK$}\label{sec:majorana}
\Cref{sec:pseudo_lindblad} reduces the spectral-gap estimate for the Gibbs sampler to two bounds: a lower bound for the unfiltered pseudo-Lindbladian $\cK$, and an upper bound for the block operator $\mathsf{M}(\wt\gamma,\wt\gamma^\dagger)$ controlling the effect of Gaussian filtering. 
We now obtain the first bound from the Jacobian.

Recall that the unfiltered jumps are obtained by conjugating the Majoranas with the Gibbs factors:
\begin{align}
    \wt{\gamma_i} = e^{-\beta H/4} \gamma_i e^{\beta H/4}, \quad i= 1,\ldots, N.
\end{align}
This conjugation preserves their mutual anticommutation relations since 
\begin{align}
    \{\wt{\gamma}_i, \wt{\gamma}_j\} = e^{-\beta H/4} \{\gamma_i, \gamma_j\} e^{\beta H/4} = 2\delta_{ij}I.
\end{align}
However, this conjugation generally does not preserve Hermiticity, since 
\begin{align}
    \wt{\gamma}_i^\dagger = e^{\beta H/4} \gamma_i e^{-\beta H/4}.
\end{align}
The mixed anticommutators are the Jacobian under a similarity transformation:
\begin{align}\label{eq:sim}
    \frac{1}{2}\{\wt \gamma_i, \wt\gamma_j^\dagger\} = e^{\beta H/4}(D\gamma)_{ij}(\beta/2) e^{-\beta H/4}.
\end{align}
The block matrix on the left is Hermitian and has the same spectrum as $D\gamma(\beta/2)$ and thus a bound on $\norm{D\gamma(\beta/2)-I}$ controls all the mixed anticommutators together.

To make use of these relations, we introduce auxiliary maps obtained by inserting operator parity into the right multiplication term of $c_i$.
At $\beta = 0$, these maps have a particularly simple action, which simply removes a Majorana from a monomial when it is present and give zero when it is absent. 
This will help us identify their common kernel and derive a gap bound.

\begin{lemma}\label{lem:aux_majorana_map}
    Consider the maps
    \begin{align}
        b_i(X) = \frac{1}{2} \left(\wt{\gamma}_i X - \cP(X) \wt{\gamma}_i^\dagger\right).
    \end{align}
    They satisfy 
    \begin{align}
        \{b_i, b_j\} = 0, \quad \{b_i, b_j^\dagger\}(X) = \frac{1}{4}\lr{\{\wt\gamma_i,\wt\gamma_j^\dagger\} X + X\{\wt\gamma_j, \wt\gamma_i^\dagger\}}.
    \end{align}
    Moreover, 
    \begin{align}
        \bigcap_{i=1}^N \ker b_i = \operatorname{span}\{\sqrt{\rho}\}.
    \end{align}
\end{lemma}
\begin{proof}
    Since $H$ is even, both $\wt{\gamma}_i$ and $\wt{\gamma}_i^\dagger$ are odd.
    Consequently, $\cP$ anticommutes with their left and right multiplication maps, and
    \begin{align}
        b_i^\dagger(X) &= \frac{1}{2}\lr{\wt\gamma_i^\dagger X + \cP(X) \wt\gamma_i}\\
        \{b_i, b_j\}(X) &= \frac{1}{4}\lr{\{\wt\gamma_i, \wt\gamma_j\} X - X \{\wt\gamma_i^\dagger, \wt\gamma_j^\dagger\}} = 0\\
        \{b_i, b_j^\dagger\}(X) &= \frac{1}{4}\lr{\{\wt\gamma_i, \wt\gamma_j^\dagger\}X + X\{\wt\gamma_j, \wt\gamma_i^\dagger\}}.
    \end{align}
    To identify the common kernel, observe that the Gibbs factors commute with physical parity, so 
    \begin{align}
        b_i\left(e^{-\beta H/4}X e^{-\beta H/4}\right) = \frac{1}{2} e^{-\beta H/4} \left(\gamma_i X - \cP(X)\gamma_i\right) e^{-\beta H/4}.
    \end{align}
    For a Majorana monomial $\gamma_S$, 
    \begin{align}
        \frac{1}{2} (\gamma_i \gamma_S - \cP(\gamma_S)\gamma_i) = 
        \begin{cases}
            \gamma_i \gamma_S, &i\in S\\
            0, &i\notin S.
        \end{cases}
    \end{align}
    For fixed $i$, the nonzero outputs have distinct supports $S \backslash \{i\}$, so they cannot cancel in a linear combination.
    Therefore, vanishing for every $i$ then forces every coefficient of a nonempty Majorana monomial to be zero.
    The common kernel of the maps inside the Gibbs factors is therefore $\operatorname{span}\{I\}$.
    The transformation $X \to e^{-\beta H/4} X e^{-\beta H/4}$ is invertible and sends $I$ to $e^{-\beta H/2}.$
    Therefore, 
    \begin{align}
        \bigcap_{i}\ker b_i = \operatorname{span}\{e^{-\beta H/2}\} = \operatorname{span}\{\sqrt{\rho}\}.
    \end{align}
\end{proof}
We can now turn these quantities into the required gap bound.
The Jacobian bound ensures that the mixed anticommutators remain uniformly positive. 
Together with $\{b_i, b_j\} = 0$, this excludes small positive eigenvalues of the associated sum of squared norms. 
\begin{lemma}[Jacobian bound implies unfiltered gap]\label{lem:curvature_implies_gap}
If $\|D\gamma(\beta/2)- I\| \leq \frac{1}{4},$ then for every $X \in \cB(\cH)$ satisfying $\langle\sqrt{\rho}, X \rangle_{\text{HS}} = 0$, 
\begin{align}
    \langle X, \cK(X) \rangle_{\text{HS}} \geq \frac{3}{4} \|X\|_2^2.
\end{align}
In particular, $\ker \cK = \operatorname{span}\{\sqrt{\rho}\}$ and $\operatorname{gap}(\cK) \geq \frac{3}{4}$. 
\end{lemma}
\begin{proof}
    By \eqref{eq:sim} and the Hermiticity of the block matrix $\left[\{\wt\gamma_i, \wt\gamma_j^\dagger\}\right]$, we can bound the spectrum of the following block matrix:
    \begin{align}
        {\rm spec}\lr{\left[\frac{1}{2}\{\wt\gamma_i,\wt\gamma_j^\dagger\}\right]} = {\rm spec}(D\gamma(\beta/2)) \subseteq [3/4, 5/4], \quad \left[\frac{1}{2}\{\wt\gamma_i,\wt\gamma_j^\dagger\}\right] \geq \frac{3}{4}I.
    \end{align}
    Applying the bound to the columns of $X_i$ and of $X_i^\dagger$, \Cref{lem:aux_majorana_map} gives
    \begin{align}
        \sum_{i,j} \langle X_i, \{b_i, b_j^\dagger\}(X_j)\rangle_{\rm HS} = \frac{1}{4}\sum_{i,j}\lr{\Tr[X_i^\dagger\{\wt\gamma_i,\wt\gamma_j^\dagger\}X_j] + \Tr[X_i\{\wt\gamma_i,\wt\gamma_j^\dagger\}X_j^\dagger]} \geq \frac{3}{4} \sum_i \norm{X_i}_2^2.\label{eq:tuple_inequality}
    \end{align}
    We first use this inequality to bound the following positive auxiliary operator
    \begin{align}\label{eq:defn_K_tilde}
        \wt{\cK} = \sum_{i} b_i^\dagger b_i,\quad \langle X, \wt{\cK}(X)\rangle_{\text{HS}} = \sum_i \|b_i(X)\|_2^2.
    \end{align}
    The relations $b_i b_j = -b_j b_i$ give 
    \begin{align}
        \wt{\cK}^2 = \sum_{i,j} b_i^\dagger \left(\{b_i, b_j^\dagger\} - b_j^\dagger b_i\right) b_j = \sum_{i,j} b_i^\dagger \{b_i, b_j^\dagger\} b_j + \sum_i b_i^\dagger \wt{\cK} b_i.
    \end{align}
    Combining \eqref{eq:tuple_inequality} and positivity of $\wt{\cK}$, we can then get 
    \begin{align}
        \langle X, \wt{\cK}^2(X) \rangle_{\text{HS}} \geq \frac{3}{4} \sum_i \|b_i(X)\|_2^2 = \frac{3}{4} \langle X, \wt{\cK}(X) \rangle_{\text{HS}}.
    \end{align}
    For an eigenvector of $\wt{\cK}$ with eigenvalue $\lambda > 0$, this inequality becomes $\lambda^2 \geq \frac{3}{4}\lambda$, and hence $\lambda \geq 3/4$.
    \Cref{lem:aux_majorana_map} and the sum-of-squares definition in \eqref{eq:defn_K_tilde} imply 
    \begin{align}
        \ker \wt{\cK} = \bigcap_i \ker b_i = \operatorname{span} \{\sqrt{\rho}\}.
    \end{align}
    Expanding in an orthonormal eigenbasis then gives
    \begin{align}
        \langle X, \wt{\cK}(X)\rangle_{\text{HS}} \geq \frac{3}{4} \|X\|_2^2,
    \end{align}
    for $X \perp \sqrt{\rho}$.
    It remains to transfer this bound to $\cK$.
    For even $X$, the definitions give $b_i(X) = c_i(X)$, so the bound follows immediately.
    For odd $X$, multiplication by physical parity gives
    \begin{align}
        b_i(XP_f) = c_i(X) P_f.
    \end{align}
    The operator $XP_f$ is also odd, and is therefore orthogonal to the even operator $\sqrt{\rho}$.
    Since multiplication by $P_f$ preserves the Hilbert-Schmidt norm, 
    \begin{align}
        \langle X, \cK(X)\rangle_{\text{HS}} &= \sum_i \|c_i(X)\|_2^2  = \sum_i \|b_i(XP_f)\|_2^2 \geq \frac{3}{4}\|XP_f\|_2^2 = \frac{3}{4} \|X\|_2^2.
    \end{align}
    Finally, every $X\perp \sqrt{\rho}$ decomposes into $(X+\cP(X))/2$ and $(X-\cP(X))/2$, both orthogonal to $\sqrt{\rho}$.
    Each $c_i$ reverses operator parity, so its outputs on these two components are orthogonal.
    Adding their bounds gives 
    \begin{align}
        \langle X, \cK(X)\rangle_{\text{HS}} \geq \frac{3}{4} \left(\norm{\frac{X+\cP(X)}{2}}_2^2 + \norm{\frac{X-\cP(X)}{2}}_2^2\right) = \frac{3}{4} \|X\|_2^2.
    \end{align}
    Since $\cK(\sqrt{\rho}) = 0$, this also identifies its kernel and proves the asserted spectral gap.
\end{proof}

\section{Bounding the SYK Jacobian with a classical Jacobian}\label{sec:jacob}
We have now reduced the spectral-gap argument to showing
\begin{align}
    \norm{D\gamma(\beta/2)-I} \leq \frac{1}{4}, \qquad \norm{\sM(\gamma(\beta/4),\gamma(-\beta/4))} \leq 150J.
\end{align}
The main obstacle is showing that the above norms do not grow with system size (at small constant $\beta$) although $\norm{H}$ is extensive. Cancellations due to the Gaussian disorder must be accounted for. We do this by analyzing a classical Hamiltonian defined on commuting variables $x_1,\dots,x_N$ with the same disorder,
\begin{align}
    G = \sum_{A\in\binom{[N]}{4}} J_A x_A.
\end{align}
We will later evaluate this on Ising configurations $\sigma\in\{\pm1\}^N$.

In the fermionic setting, recall that we set $\gamma_j(s) = e^{-s \ad_H}\gamma_j$. Observe that the action of the commutator on a fermionic string is
\begin{align}\label{eq:hgam}
    [H, \gamma_{w_1}\cdots\gamma_{w_L}] = -2\sum_{a=1}^L \sum_{A\ni w_a} (-1)^{|\{b\in A:b < w_a\}|} J_A \gamma_{w_1}\cdots \gamma_{w_{a-1}} \gamma_{A\setminus\{w_a\}} \gamma_{w_{a+1}} \cdots \gamma_{w_L},
\end{align}
where the sign is due to reordering fermions.
Classically, this operation is instead achieved by
\begin{align}
    (2\nabla G \cdot \nabla) x_{w_1}\cdots x_{w_L} &= 2\sum_{j=1}^N (\partial_j G) \partial_j (x_{w_1}\cdots x_{w_L} )\\
    &= 2\sum_{a=1}^L \sum_{A \ni w_a} J_A x_{w_1} \cdots x_{w_{a-1}} x_{A\setminus\{w_a\}} x_{w_{a+1}} \cdots x_{w_L}.\label{eq:hx}
\end{align}
Hence, the analogue of $\gamma_j(s)$ will be $x_j(s) = e^{s(2\nabla G \cdot \nabla)} x_j\Bigg|_{x=\sigma}$ for some initial condition $\sigma$ that we will make explicit. Equivalently, $x$ satisfies
\begin{align}
    \frac{d}{ds} x(s;\sigma) = 2\nabla G(x(s;\sigma)), \qquad x(0;\sigma) = \sigma.
\end{align}
We then expand the Jacobians
\begin{align}
    D\gamma(s) = I + \sum_{m \geq 1} s^m F_m, \qquad Dx(s;\sigma) = I + \sum_{m \geq 1} s^m \wt F_m(\sigma),
\end{align}
where the classical Jacobian is evaluated at $\sigma$:
\begin{align}\label{eq:fm}
    (F_m)_{ji} = \frac{(-1)^m}{2m!} \{\gamma_i, \ad_H^m \gamma_j\}, \qquad (\wt F_m(\sigma))_{ji} = \frac{1}{m!} \partial_{x_i} \lr{2\nabla G \cdot \nabla}^m x_j \Bigg|_{x=\sigma}.
\end{align}
Our main lemma bounds the fermionic Jacobian by the classical Jacobian. Note that $F_m$ is an $N\times N$ matrix whose entries are $d \times d$ operators; below, we use $\Tr$ to denote the full trace on $\C^N \otimes \C^d$. On the other hand, $\wt F_m$ is an $N\times N$ matrix. This dimension mismatch introduces the factor of $d$ in the lemma below.
\begin{lemma}\label{lem:jacob}
    For integers $m,p\geq1$, the coefficients in \eqref{eq:fm} satisfy
    \begin{align}
        \frac{1}{d}\E_J\Tr[(F_m^\dagger F_m)^p]
        \leq\E_J\Tr[(\wt F_m(\mathbf1)^\dagger\wt F_m(\mathbf1))^p].
    \end{align}
\end{lemma}
\begin{proof}
    To easily evaluate the expectation over Gaussian disorder, we group terms by the Hamiltonian coefficients $J_A$. Let $\alpha_A$ record the number of times each $J_A$ appears; we write $J^\alpha = \prod_A J_A^{\alpha_A}$ and $|\alpha| = \sum_A \alpha_A$. Using
    \begin{align}
        \partial_i(\gamma_{w_1}\cdots\gamma_{w_L}) = \sum_{a=1}^L (-1)^{a-1} \delta_{i,w_a} \gamma_{w_1}\cdots\gamma_{w_{a-1}}\gamma_{w_{a+1}}\cdots\gamma_{w_L},
    \end{align}
    we have for $m=|\alpha|$ by applying the triangle inequality on \eqref{eq:hgam} and comparing the terms to \eqref{eq:hx} (noting the norm of a Majorana string is 1) that
    \begin{align}
        \left\|
            [J^\alpha]\partial_i\ad_H^m\gamma_j
        \right\|
        &\leq
        [J^\alpha]\partial_{x_i}
        (2\nabla G\cdot\nabla)^m x_j
        \bigg|_{x=\mathbf 1}.
    \end{align}
    Since $\ad_H^m\gamma_j$ is odd, \eqref{eq:fm} then gives
    \begin{align}
        \norm{[J^\alpha](F_m)_{ji}}
        &=
        \frac{1}{m!}
        \left\|
            [J^\alpha]\partial_i\ad_H^m\gamma_j
        \right\| \leq
        \frac{1}{m!}
        [J^\alpha]\partial_{x_i}
        (2\nabla G\cdot\nabla)^m x_j
        \bigg|_{x=\mathbf 1} =
        [J^\alpha](\wt F_m(\mathbf 1))_{ji}.
    \end{align}
    Each entry of $F_m$ is homogeneous of degree $m$ in the couplings, and
    $\E_J J^{\alpha_1+\cdots+\alpha_{2p}}\geq0$.
    With $i_{p+1}=i_1$, we obtain
    \begin{align}\label{eq:gaussian_trace_expansion}
        &\frac{1}{d}\E_J\Tr[(F_m^\dagger F_m)^p]\notag\\
        &=\sum_{i,j\in[N]^p}
          \sum_{|\alpha_1|=\cdots=|\alpha_{2p}|=m}
          \E_J J^{\sum_r\alpha_r}\,
          \frac{1}{d}\operatorname{Re}\Tr\prod_{r=1}^p
          \left(\bigl([J^{\alpha_{2r-1}}](F_m)_{j_ri_r}\bigr)^\dagger
          [J^{\alpha_{2r}}](F_m)_{j_ri_{r+1}}\right)\notag\\
        &\leq\sum_{i,j\in[N]^p}
          \sum_{|\alpha_1|=\cdots=|\alpha_{2p}|=m}
          \E_J J^{\sum_r\alpha_r}\,
          \prod_{r=1}^p
          [J^{\alpha_{2r-1}}](\wt F_m(\mathbf1))_{j_ri_r}
          [J^{\alpha_{2r}}](\wt F_m(\mathbf1))_{j_ri_{r+1}}\\
        &=\E_J\Tr[(\wt F_m(\mathbf1)^T\wt F_m(\mathbf1))^p].\notag
    \end{align}
    The inequality uses $|\Tr Z|\leq d\|Z\|$, submultiplicativity of the
    operator norm, and the coefficient bound above.
\end{proof}

Since it is commuting, the classical Jacobian can now be easily bounded with standard techniques; the proof is in \Cref{app:classical}. This yields an immediate bound on the SYK Jacobian.
\begin{lemma}\label{lem:classical}
    For integers $m, p \geq 1$ with $mp \leq 7N$,
    \begin{align}
        \lr{\E_J \norm{\wt F_m(\mathbf 1)}^{2p}}^{1/(2p)} \leq (105J)^m.
    \end{align}
\end{lemma}
\begin{corollary}\label{cor:fbound}
    For integers $m, p \geq 1$ with $mp \leq 7N$, the operators in \eqref{eq:fm} satisfy
    \begin{align}
        \E \norm{F_m}^{2p} \leq Nd(105 J)^{2pm}.
    \end{align}
\end{corollary}
\begin{proof}
    By \Cref{lem:jacob,lem:classical},
    \begin{align}
        \E\|F_m\|^{2p} \leq\E\Tr[(F_m^\dagger F_m)^p] \leq d\E\Tr[(\wt F_m(\mathbf 1)^\dagger\wt F_m(\mathbf 1))^p] \leq Nd\E\|\wt F_m(\mathbf 1)\|^{2p}
        \leq Nd(105J)^{2pm}.
    \end{align}
\end{proof}
A similar argument also controls $\mathsf M$. Recalling from \eqref{eq:mprop} that
\begin{align}
    \partial_s (D\gamma)_{ij}(s) = \sum_a \sM(\gamma(s), -\gamma(s))_{ia}\lr{(D\gamma)_{aj}(s)},
\end{align}
we observe that the analogous classical equation is
\begin{align}\label{eq:classical_variation}
    \partial_s Dx(s;\sigma) = 2\nabla^2 G(x(s;\sigma)) Dx(s;\sigma).
\end{align}
At $\sigma = (1,\dots,1)$, both $2\nabla^2 G(x(s;\mathbf 1))$ and $Dx(s;\mathbf 1)-I$ have nonnegative coefficients in $s$ and the disorder $J_A$. Consequently, the Hessian is bounded coefficient-wise by $\partial_s Dx(s;\mathbf 1)$. This allows us to show the following result analogous to \Cref{lem:jacob}, as well as a corollary similar to \Cref{cor:fbound}.
\begin{lemma}\label{lem:m_jacobian}
    For an integer $m\geq0$, let $M_m=[s^m]\sM(\gamma(s),\gamma(-s))$.
    Then for every integer $p\geq1$,
    \begin{align}
        \frac{1}{d^2}\E_J\Tr[(M_m^\dagger M_m)^p]
        \leq(m+1)^{2p}\E_J\Tr[(\wt F_{m+1}(\mathbf1)^\dagger
          \wt F_{m+1}(\mathbf1))^p].
    \end{align}
    Here the first trace is on $\C^N\otimes\cB(\cH)$ and the second is on $\C^N$.
\end{lemma}
\begin{proof}
    We use the same word expansion as in \Cref{lem:jacob}, now before taking the fermionic derivative.
    Its coefficient bound is
    \begin{align}\label{eq:word_coefficient_bound}
        \|[s^rJ^\alpha]\gamma_i(\pm s)\|
        \leq[s^rJ^\alpha]x_i(s;\mathbf1),
        \qquad |\alpha|=r.
    \end{align}
    The right side is nonnegative. For each term in \eqref{eq:mdef}, use
    $\|X\mapsto AXB\|_{2\to2}\leq\|A\|\|B\|$.
    For $i\neq j$, the three terms and the factor $1/3$ give
    \begin{align}\label{eq:m_hessian_coefficients}
        \|[J^\alpha](M_m)_{ij}\|_{2\to2}
        &\leq[s^mJ^\alpha]
          \sum_{\substack{k,l\\i,j,k,l\ \mathrm{distinct}}}
          J_{\{i,j,k,l\}}x_k(s;\mathbf1)x_l(s;\mathbf1) =[s^mJ^\alpha]\bigl(2\nabla^2G(x(s;\mathbf1))\bigr)_{ij}.
    \end{align}
    For $i=j$, both sides of \eqref{eq:m_hessian_coefficients} vanish.
    Every coefficient of $x_i(s;\mathbf1)$, $Dx(s;\mathbf1)-I$, and $2\nabla^2G(x(s;\mathbf1))$ is nonnegative.
    Thus \eqref{eq:classical_variation} gives, for $|\alpha|=m+1$,
    \begin{align}\label{eq:m_derivative_coefficients}
        [s^mJ^\alpha]\bigl(2\nabla^2G(x(s;\mathbf1))\bigr)_{ij}
        &\leq[s^mJ^\alpha]
             \bigl(2\nabla^2G(x(s;\mathbf1))Dx(s;\mathbf1)\bigr)_{ij} =[s^mJ^\alpha]\partial_s(Dx)_{ij}(s;\mathbf1) \\
        &=(m+1)[J^\alpha](\wt F_{m+1}(\mathbf1))_{ij}.
    \end{align}
    Combining \eqref{eq:m_hessian_coefficients} and \eqref{eq:m_derivative_coefficients},
    \begin{align}\label{eq:m_coefficient_comparison}
        \|[J^\alpha](M_m)_{ij}\|_{2\to2}
        \leq(m+1)[J^\alpha](\wt F_{m+1}(\mathbf1))_{ij}.
    \end{align}
    Each block of $M_m$ is homogeneous of degree $m+1$ in the couplings.
    Apply the trace expansion \eqref{eq:gaussian_trace_expansion} with blocks
    acting on the $d^2$-dimensional Hilbert-Schmidt space.
    Using \eqref{eq:m_coefficient_comparison} and nonnegativity of the Gaussian moments,
    \begin{align}
        &\frac{1}{d^2}\E_J\Tr[(M_m^\dagger M_m)^p]\notag\\
        &\leq(m+1)^{2p}\sum_{i,j\in[N]^p}
          \sum_{|\alpha_1|=\cdots=|\alpha_{2p}|=m+1}
          \E_J J^{\sum_r\alpha_r}\,
          \prod_{r=1}^p
          [J^{\alpha_{2r-1}}](\wt F_{m+1}(\mathbf1))_{j_ri_r}
          [J^{\alpha_{2r}}](\wt F_{m+1}(\mathbf1))_{j_ri_{r+1}}\notag\\
        &=(m+1)^{2p}\E_J\Tr[(\wt F_{m+1}(\mathbf1)^T
          \wt F_{m+1}(\mathbf1))^p],
    \end{align}
    where $i_{p+1}=i_1$.
\end{proof}
\begin{corollary}\label{cor:mbound}
    For integers $m \geq 0, p \geq 1$ with $(m+1)p \leq 7N$,
    \begin{align}
        \E \norm{M_m}^{2p} \leq Nd^2 (m+1)^{2p} (105J)^{2p(m+1)}.
    \end{align}
\end{corollary}
\begin{proof}
    By \Cref{lem:m_jacobian,lem:classical},
    \begin{align}
        \E\|M_m\|^{2p}
        \leq\E\Tr[(M_m^\dagger M_m)^p]
        \leq Nd^2(m+1)^{2p}\E\|\wt F_{m+1}(\mathbf 1)\|^{2p}
        \leq Nd^2(m+1)^{2p}(105J)^{2p(m+1)}.
    \end{align}
\end{proof}
Finally, with these moment bounds in hand, we can prove \Cref{thm:main_result} with standard concentration inequalities. We use an elementary bound on the SYK operator norm.

\begin{fact}[SYK norm]\label{fac:norm}
    The SYK Hamiltonian satisfies $\Pr[\norm{H} > 2JN] \leq 2^{-5N/2}$.
\end{fact}
\begin{proof}
    The bound $\E \Tr H^{2N} \leq d \E \lr{\sum_A J_A}^{2N} = d(2N-1)!! \lr{\sum_A \E J_A^2}^N \leq d(2N)^N (J^2N/4)^N$, where the last inequality used  $\sum_A \E J_A^2 = \binom{N}{4}(6J^2/N^3) \leq J^2N/4$, implies that $\Pr[\norm{H} > 2JN] \leq \E \Tr H^{2N}/(2JN)^{2N} \leq 2^{-5N/2}$.
\end{proof}

This fact and the above bounds are used to show the required properties of $D\gamma$ and $\sM$ for fast mixing.

\begin{corollary}[High-probability Jacobian bounds]\label{cor:jacobian_high_probability}
    With probability at least $1-2e^{-N}$ over the disorder, $\|H\|\leq2JN$ and, simultaneously for every real $s$ with $133J|s|<1$,
    \begin{align}\label{eq:uniform_jacobian_bounds}
        \|D\gamma(s)-I\|&\leq\frac{133J|s|}{1-133J|s|},
        \qquad \|\mathsf M(\gamma(s),\gamma(-s))\|
        \leq\frac{133J}{(1-133J|s|)^2}.
    \end{align}
\end{corollary}
\begin{proof}
    For each $1\leq m\leq N$, take $p=\lceil6N/m\rceil$, so $6N\leq pm\leq7N$.
    By \Cref{cor:fbound,cor:mbound} and Markov's inequality,
    \begin{align}
        \Pr\left[\exists m\in[N]:
          \max\{\|F_m\|,\|M_{m-1}\|/m\}>(133J)^m\right] \leq N^2(d+d^2)(105/133)^{12N}
         \leq2\cdot4^{-N}.
    \end{align}
    Intersect this event's complement with $\{\|H\|\leq2JN\}$;
    by \Cref{fac:norm}, the latter event fails with probability at most $2^{-5N/2}$.
    For $m>N$, the bound $m!\geq(m/e)^m$ gives
    \begin{align}
        \|F_m\|
        &\leq\frac{N(2\|H\|)^m}{m!}
         \leq N\lr{\frac{4eJN}{m}}^m
         \leq N(4eJ)^m\leq(133J)^m.
    \end{align}
    The identity $\mathsf M(\gamma(s),\gamma(-s))_{ij}(X)=e^{-sH}(M_0)_{ij}(e^{sH}Xe^{sH})e^{-sH}$ and $\|M_0\|\leq133J$ imply, for $r\geq N$,
    \begin{align}
        \|M_r\|
        &\leq\|M_0\|\sum_{a_1+\cdots+a_4=r}
          \frac{\|H\|^r}{a_1!\cdots a_4!}
         =\|M_0\|\frac{(4\|H\|)^r}{r!} \leq133J\lr{\frac{8eJN}{r}}^r
         \leq(r+1)(133J)^{r+1}.
    \end{align}
    Thus, on a single event of probability at least
    $1-2\cdot4^{-N}-2^{-5N/2}\geq1-2e^{-N}$,
    \begin{align}
        \|F_m\|\leq(133J)^m,
        \qquad \|M_{m-1}\|\leq m(133J)^m,
        \qquad m\geq1.
    \end{align}
    Summing the series gives
    \begin{align}
        \|D\gamma(s)-I\|
        &\leq\sum_{m\geq1}(133J|s|)^m
         =\frac{133J|s|}{1-133J|s|},\\
        \|\mathsf M(\gamma(s),\gamma(-s))\|
        &\leq\sum_{m\geq1}|s|^{m-1}m(133J)^m
         =\frac{133J}{(1-133J|s|)^2}.
    \end{align}
\end{proof}

\begin{proof}[Proof of \Cref{thm:main_result}]
    On the event in \Cref{cor:jacobian_high_probability}, for $0\leq\beta J\leq1/1000$,
    \begin{align}
        \|D\gamma(\beta/2)-I\|
        &\leq\frac{133\beta J/2}{1-133\beta J/2}
         <\frac{1}{4},\\
        \|\mathsf M(\wt\gamma,\wt\gamma^\dagger)\|
        &\leq\frac{133J}{(1-133\beta J/4)^2}
         <150J.
    \end{align}
    Applying \Cref{lem:curvature_implies_gap,cor:pseudo_comparisons,lem:filtered_to_lind_gap},
    \begin{align}
        \operatorname{gap}(\cL)
        &\geq2\operatorname{gap}(\cK_f)
         \geq2\lr{\frac{13}{14}}^2\operatorname{gap}(\cK)
         \geq\frac{507}{392}>1,\\
        t_{\mathrm{mix}}(\epsilon)
        &\leq\frac{392}{507}\lr{\frac{N}{4}\log2+2\beta JN+\log\frac{1}{\epsilon}}
         =O\lr{N+\log\frac{1}{\epsilon}}.
    \end{align}
\end{proof}

\paragraph{Acknowledgment.} The authors thank Mathews Boban, Shayan Oveis Gharan, Jake Hofgard, Petar Nizic-Nikolac and Robert Wang for discussions on classical results obtained by Bochner inequalities. YC is supported in part by the Air Force Office of Scientific Research under grant agreement FA9550-21-1-0392, by the Stanford Graduate Fellowship in Science \& Engineering, and by the National Science Foundation Graduate Research Fellowship Program under Grant No. DGE-2146755.YZ acknowledges funding from the
National Science Foundation (PHY-1733907). The Institute for Quantum Information and Matter is an NSF Physics Frontiers Center. AZ is supported by a Hertz Fellowship and a grant from the Simons Foundation (MP-SIP-00001553, AWH).

\paragraph{AI disclosure.} The main idea behind the proven theorem statement was proposed by ChatGPT 6. The authors developed and substantially simplified the argument, introducing the comparison between fermionic and classical Jacobians that is central to the proof presented here. All proofs were checked by the authors, who take responsibility for their correctness.

\bibliographystyle{alpha}
\bibliography{ref}

\newpage
\appendix 
\section{Jacobian bounds for the classical Hamiltonian}\label{app:classical}
The proof is short using three facts from the literature. All vector norms are Euclidean.

\begin{fact}[Banach's theorem, Eq.~(23) of \cite{friedland2016nuclear}]
    For an integer $r\geq1$ and a symmetric $r$-linear form $T:(\mathbb R^n)^r\to\mathbb R$,
    \begin{align}
        \sup_{\|u_1\|=\cdots=\|u_r\|=1}|T(u_1,\ldots,u_r)|
        =\sup_{\|u\|=1}|T(u,\ldots,u)|.
    \end{align}
\end{fact}
\noindent Consequently, every homogeneous polynomial map
$P:\mathbb R^n\to\mathbb R^a$ of degree $r\geq1$ satisfies
\begin{align}\label{eq:homogeneous_derivative_bound}
    \sup_{\|x\|=1}\|DP(x)\|
    \leq r\sup_{\|x\|=1}\|P(x)\|.
\end{align}
Indeed, for a unit vector $w$, write
$\langle w,P(x)\rangle=T_w(x,\ldots,x)$ with $T_w$ symmetric.
For unit $x,h,w$, Banach's theorem gives
\begin{align}
    |\langle w,DP(x)h\rangle|
    &=r|T_w(h,x,\ldots,x)|
    \leq r\sup_{\|z\|=1}|\langle w,P(z)\rangle|
    \leq r\sup_{\|z\|=1}\|P(z)\|.
\end{align}

\begin{fact}[Sudakov-Fernique, Theorem~2.2.3 of~\cite{adler2007random}]
    Let $(X_t)_{t\in T}$ and $(Y_t)_{t\in T}$ be centered real Gaussian
    processes with continuous sample paths on a compact metric space $T$.
    If $\mathbb E(X_t-X_u)^2\leq\mathbb E(Y_t-Y_u)^2$ for all $t,u\in T$,
    then
    \begin{align}
        \mathbb E\sup_{t\in T}X_t\leq\mathbb E\sup_{t\in T}Y_t.
    \end{align}
\end{fact}

\begin{fact}[Borell-TIS, Theorem~2.1.1 of~\cite{adler2007random}]
    Let $(X_t)_{t\in T}$ be a centered real Gaussian process with continuous
    sample paths on a compact metric space $T$, and let
    $v^2=\sup_{t\in T}\mathbb E X_t^2>0$.
    Then $\mathbb E\sup_{t\in T}X_t$ is finite and, for $u\geq0$,
    \begin{align}
        \Pr\left[\sup_{t\in T}X_t>
          \mathbb E\sup_{t\in T}X_t+u\right]
        \leq e^{-u^2/(2v^2)}.
    \end{align}
\end{fact}

\begin{proof}[Proof of \Cref{lem:classical}]
    Put $\ell(G)=\sup_{\|x\|=1}|G(x)|$, which equals the multilinear norm
    by Banach's theorem. For real unit vectors $u,v$,
    \begin{align}
        \mathbb E_J(G(u)-G(v))^2
        &=\frac{6J^2}{N^3}\sum_{A\in\binom{[N]}4}(u_A-v_A)^2 \leq\frac{J^2}{4N^3}\|u^{\otimes4}-v^{\otimes4}\|^2
         =\frac{J^2}{2N^3}(1-\langle u,v\rangle^4)
         \leq\frac{J^2}{N^3}\|u-v\|^2.
    \end{align}
    Sudakov-Fernique compares $G(u)$ with
    $JN^{-3/2}\langle g,u\rangle$, where $g\sim\mathcal N(0,I_N)$.
    Since $G(e_1)=0$ and $G$ and $-G$ have the same law,
    \begin{align}
        \mathbb E_J\ell(G)
        &\leq2\mathbb E_J\sup_{\|u\|=1}G(u)
         \leq\frac{2J}{N^{3/2}}\mathbb E\|g\|
         \leq\frac{2J}{N},\qquad \sup_{\|u\|=1}\mathbb E_JG(u)^2
        \leq\frac{J^2}{4N^3}.
    \end{align}
    Apply Borell-TIS to the continuous Gaussian process
    $(\eta,u)\mapsto\eta G(u)$ on $\{\pm1\}\times S^{N-1}$.
    For $t\geq0$, we have $\Pr[\ell(G)>2J/N+t]\leq e^{-2N^3t^2/J^2}$. Integrating this tail gives, for every integer $k\geq1$,
    \begin{align}\label{eq:tensor_moment}
        \bigl(\mathbb E_J\ell(G)^{2k}\bigr)^{1/(2k)}
        &\leq\frac{2J}{N}
          +\left(2k\int_0^\infty t^{2k-1}e^{-2N^3t^2/J^2}\,dt\right)^{1/(2k)} \hspace{-0.5em} = \frac{2J}{N}+\frac{J(k!)^{1/(2k)}}{\sqrt2\,N^{3/2}}
         \leq\frac{J}{N}\left(2+\sqrt{\frac{k}{2N}}\right).
    \end{align}
    The polynomial $(2\nabla G\cdot\nabla)^m x$ has degree $2m+1$.
    Applying \eqref{eq:homogeneous_derivative_bound} to $G$ and then
    successively to these polynomials gives
    \begin{align}
        \sup_{\|x\|=1}\|(2\nabla G\cdot\nabla)^m x\|
        &\leq8\ell(G)(2m-1)
          \sup_{\|x\|=1}\|(2\nabla G\cdot\nabla)^{m-1}x\| \leq(8\ell(G))^m(2m-1)!!,\\
        \|\widetilde F_m(\mathbf 1)\|
        &=\frac1{m!}\left\|D\bigl((2\nabla G\cdot\nabla)^m x\bigr)
            \big|_{x=\mathbf 1}\right\| \leq\frac{(8N\ell(G))^m(2m+1)!!}{m!}
         \leq(24N\ell(G))^m.
    \end{align}
    Here $\|\mathbf 1\|=\sqrt N$ and
    $(2m+1)!!/m!=\prod_{r=1}^m(2+1/r)\leq3^m$.
    Taking $k=mp\leq7N$ in \eqref{eq:tensor_moment},
    \begin{align}
        \bigl(\mathbb E_J
             \|\widetilde F_m(\mathbf 1)\|^{2p}\bigr)^{1/(2p)}
        &\leq(24N)^m\bigl(\mathbb E_J\ell(G)^{2mp}\bigr)^{1/(2p)} \leq\left[24J\left(2+\sqrt{\frac{mp}{2N}}\right)\right]^m
        \leq(105J)^m.
    \end{align}
\end{proof}

\end{document}